\documentclass[10pt,a4paper,english]{article}

\usepackage[utf8]{inputenc}
\usepackage[a4paper]{geometry}
\usepackage{enumerate,hyperref}
\usepackage[inline]{enumitem}
\usepackage{amsmath,amsthm,amsfonts,amssymb,stmaryrd,mathtools,centernot}
\usepackage{etex,etoolbox,thmtools,environ}
\usepackage{xspace}
\usepackage{tikz}
\usetikzlibrary{calc,matrix,arrows,positioning,intersections}

\newif\ifdraft
\newif\ifappendix

\appendixtrue

\newcommand{\suppressed}[1]{
	\PackageWarning{}{#1 suppressed on page \thepage\space on input line \number\inputlineno!}}

\ifdraft
	\PackageWarning{}{Draft macros enabled.}
\else
	\renewcommand{\marginpar}[1]{\suppressed{Margin paragraph}}
\fi

\ifappendix
	\newcommand{\assumeapx}[1]{#1}
\else
	\newcommand{\assumeapx}[1]{\suppressed{Content assuming the Appendix}}
\fi

\ifdraft
	\usepackage[breakable]{tcolorbox}
	\newenvironment{ccbox}[2][]{
			\begin{tcolorbox}[colback=#2!20,colframe=#2,
				breakable,
				boxsep=3pt,boxrule=.5pt,#1]
			}{\end{tcolorbox}}

\else
	\NewEnviron{ccbox}[2][]{}
	\NewEnviron{alertbox}[1][]{}	
	\NewEnviron{commentbox}[1][]{}
	\NewEnviron{highlightbox}[1][]{}
	\NewEnviron{draftbox}[1][]{}
	\NewEnviron{todobox}[1][]{}
\fi

\theoremstyle{plain}\newtheorem{theorem}{Theorem}[section]
\theoremstyle{plain}\newtheorem{corollary}[theorem]{Corollary}
\theoremstyle{plain}\newtheorem{proposition}[theorem]{Proposition}
\theoremstyle{plain}\newtheorem{lemma}[theorem]{Lemma}
\theoremstyle{plain}\newtheorem{definition}[theorem]{Definition}
\theoremstyle{remark}\newtheorem{remark}[theorem]{Remark}
\theoremstyle{remark}\newtheorem{example}[theorem]{Example}

\makeatletter
\ifcsname spn@wtheorem\endcsname
	\providecommand{\@lastarg}[1]{#1}
\else
	\providecommand{\@lastarg}[4]{#4}
\fi
\newcounter{freshthmlabel}

\newif\ifnoproofsketch
\noproofsketchfalse

\newcommand{\setthmlabel}{%
	\stepcounter{freshthmlabel}%
	\label{proofatend:\thefreshthmlabel}%
	\global\noproofsketchtrue
}

\newcommand\fixstatement[2][\proofname\space of]{%
  \ifcsname thmt@original@#2\endcsname
    \AtEndEnvironment{#2}{%
      \xdef\pat@label{\expandafter\expandafter\expandafter
        \@lastarg\csname thmt@original@#2\endcsname}
		\xdef\pat@proofof{\@nameuse{pat@proofof@#2}}%
		\setthmlabel
    }%
  \else
	\ifcsname #2name\endcsname
		\AtEndEnvironment{#2}{%
			\xdef\pat@label{\expandafter\expandafter\expandafter\@lastarg\csname #2name\endcsname}
			\xdef\pat@proofof{\@nameuse{pat@proofof@#2}}%
			\setthmlabel
		}%
	\else
		\AtEndEnvironment{#2}{%
			\xdef\pat@label{\expandafter\expandafter\expandafter\@lastarg\csname #1\endcsname}
			\xdef\pat@proofof{\@nameuse{pat@proofof@#2}}%
			\setthmlabel
		}%
    \fi
  \fi
  \@namedef{pat@proofof@#2}{#1}%
}

\globtoksblk\prooftoks{100}
\newcounter{proofcount}

\newcommand{\seeproofref}{%
	\ifappendix
	\renewcommand*{\qedsymbol}{%
		\hyperref[proofatend:proof-\thefreshthmlabel]{%
			\(\Box\)}}
	\fi
}

\NewEnviron{proofsketch}[1][Proof (Sketch)]{%
	\begin{proof}[#1]		
		\BODY
		\ifappendix
				\seeproofref
		\fi
	\end{proof}
	\global\noproofsketchfalse
}

\NewEnviron{proofatend}{%
	\ifnoproofsketch
		\ifappendix
				\begin{proof}[Proof (Omitted)]
					 See Appendix~\ref{proofatend:proof-\thefreshthmlabel}.
					 \seeproofref
				\end{proof}
		\fi
	\fi
	\edef\next{%
		\noexpand\begin{proof}[{\pat@proofof\space\pat@label\space\noexpand\ref{proofatend:\thefreshthmlabel}}]%
		\noexpand\phantomsection%
		\noexpand\label{proofatend:proof-\thefreshthmlabel}%
		\unexpanded\expandafter{\BODY}}%
		\global\toks\numexpr\prooftoks+\value{proofcount}\relax=\expandafter{\next\end{proof}}
	\stepcounter{proofcount}}
	
	\def\printproofs{%
		\ifnum0<\value{proofcount}
			\section{Omitted proofs}
			\label{proofatend:proofs}%
			\count@=\z@
			\loop
				\the\toks\numexpr\prooftoks+\count@\relax
				\ifnum\count@<\value{proofcount}%
					\advance\count@\@ne
			\repeat
		\fi
	}
\makeatother

\fixstatement{theorem}
\fixstatement{proposition}
\fixstatement{corollary}
\fixstatement{lemma}

\makeatletter
\newcommand{\superimpose}[2]{%
  {\ooalign{$#1\@firstoftwo#2$\cr\hfil$#1\@secondoftwo#2$\hfil\cr}}}
\makeatother

\makeatletter
	\newcommand{\footnoteref}[1]{%
		\protected@xdef\@thefnmark{\ref{#1}}\@footnotemark%
	}
\makeatother

\makeatletter
\newcommand{\customlabel}[2]{%
	\protected@write \@auxout {}{\string \newlabel {#1}{{#2}{\thepage}{#2}{#1}{}} }%
	\hypertarget{#1}{#2}
}
\makeatother

\newcommand{\ltrue}{\texttt{t\!t}}
\newcommand{\lfalse}{\texttt{f\!f}}

\newcommand{\defeq}{\triangleq}
\newcommand{\defiff}{\stackrel{\triangle}{\iff}}

\newcommand{\cat}[1]{\textnormal{\textsf{#1}}\xspace}
\newcommand{\set}{\cat{Set}}
\newcommand{\coalg}[1]{{#1\cat{-CoAlg}}}

\newcommand{\hole}{ {\mbox{\large\bf-}} }
\newcommand{\id}{\mathrm{Id}}
\newcommand{\f}[1]{\mathcal{F}_{#1}}
\newcommand{\ff}[1]{\f{\mathfrak{#1}}}
\newcommand{\fw}{\ff{W}}
\newcommand{\p}[1]{\mathcal{P}_{#1}}
\newcommand{\pf}{\p{\!\!f}}
\newcommand{\bottom}{\perp}

\newcommand{\support}[1]{\lfloor #1 \rfloor}
\newcommand{\totalweight}[1]{\llfloor #1 \rrfloor}

\newcommand{\totalweightg}[1]{\left\lfloor\kern-4pt\left\lfloor #1 \right\rfloor\kern-4pt\right\rfloor}

\newcommand\restr[2]{{\left.\kern-\nulldelimiterspace#1\vphantom{\big|}\right|_{#2}}}
\newcommand\corestr[2]{{\left.\kern-\nulldelimiterspace#1\vphantom{\big|}\right|^{#2}}}

\newcommand{\llbrace}{\lbrace\kern-2.2pt\vert}
\newcommand{\rrbrace}{\vert\kern-2.2pt\rbrace}

\newcommand{\pepasync}[1]{\raisebox{-1.0ex}{$\;\stackrel{\mbox{\large $\rhd\hspace{-1.2ex}\lhd$}}{\scriptscriptstyle #1}\,$}}

\newsavebox{\xtaTempBox}
\newlength{\xtaMinLen}
\newcommand{\rightarrowh}[1]{%
  \mathbin{
  \tikz[baseline=-.75ex]{
    \setlength{\xtaMinLen}{1.2em}
    \draw[#1] (0,0) -- (\xtaMinLen,0);}}}
\newcommand{\xrightarrowh}[4]{%
  \sbox{\xtaTempBox}{\hbox{\( \scriptstyle\mkern#3#2\mkern#4 \)}}
  \mathbin{
  \tikz[baseline=-0.6ex]{
    \setlength{\xtaMinLen}{1.2em}
    \setlength{\xtaMinLen}{\maxof{\wd\xtaTempBox}{\xtaMinLen}}
    \draw[#1] (0,0) --
    node[midway,above=-0.4ex]{\usebox\xtaTempBox} (\xtaMinLen,0);}}}
\newcommand{\rightarrowu}{\rightarrowh{-open triangle 60}}
\newcommand{\xrightarrowu}[1]{\xrightarrowh{-open triangle 60}{#1}{7mu}{17mu}}

\newcommand{\xrightarroww}[1]{\xrightarrowh{-triangle 60}{#1}{7mu}{17mu}}

\newcommand{\xrightarrowg}[1]{\xrightarrowh{-angle 90}{#1}{7mu}{9mu}}
\newcommand{\rightarrows}{\rightarrowh{->,double equal sign distance,-implies}}

\newcommand{\xrightarrowt}[1]{\xrightarrowh{-open triangle 60 reversed}{#1}{7mu}{17mu}}

\hypersetup{
	colorlinks=true,
	linkcolor=blue!50!black,
	citecolor=blue!50!black,
	filecolor=blue!50!black,
	urlcolor=blue!50!black,
	pdfpagelayout=SinglePage,
	pdfpagemode=UseOutlines,
	pdftitle={Structural operational semantics for non-deterministic processes with quantitative aspects},
	pdfauthor={Marino Miculan and Marco  Peressotti},
	pdfsubject={Concurrency Theory},
	pdfkeywords={rule formats, behavioural congruences, concurrent systems, quantitative models}
}

\title{
	Structural operational semantics for\\
	non-deterministic processes with quantitative aspects\thanks{
	This work is partially supported by MIUR PRIN project 2010LHT4KM, \emph{CINA}.}
}
\author{
	\begin{tabular}{ccc}
	Marino Miculan&\qquad& Marco  Peressotti\\[-.4ex]
	\small\href{mailto:marino.miculan@uniud.it}{\tt marino.miculan@uniud.it}
	&\qquad&
	\small\href{mailto:marco.peressotti@uniud.it}{\tt marco.peressotti@uniud.it}
	\end{tabular}\\[-.2ex]
	\small	Laboratory of Models and Applications of Distributed Systems \\[-.4ex]
	\small	Department of Mathematics and Computer Science\\[-.4ex]
	\small	University of Udine, Italy\\
}

\date{}

\begin{document}

\maketitle

\begin{abstract}
	Recently, unifying theories for processes combining non-determinism
	with quantitative aspects (such as probabilistic or stochastically
	timed executions) have been proposed with the aim of providing
	general results and tools.  This paper provides two contributions in
	this respect.  First, we present a general GSOS specification format
	and a corresponding notion of bisimulation for non-deterministic
	processes with quantitative aspects.  These specifications define
	labelled transition systems according to the ULTraS model, an
	extension of the usual LTSs where the transition relation associates
	any source state and transition label with \emph{state reachability
	weight functions} (like, e.g., probability distributions). This
	format, hence called \emph{Weight Function GSOS} (WF-GSOS), covers
	many known systems and their bisimulations (e.g.~PEPA, TIPP, PCSP)
	and GSOS formats (e.g.~GSOS, Weighted GSOS, Segala-GSOS).
	
	The second contribution is a characterization of these systems
	as coalgebras of a class of functors, parametric in the weight
	structure. This result allows us to prove \emph{soundness} and
	\emph{completeness} of the WF-GSOS specification format, and
	that bisimilarities induced by these specifications are always
	congruences.
\end{abstract}

\section{Introduction}


Process calculi and labelled transition systems have proved very
successful for modelling and analysing concurrent, non-deterministic
systems.
This success has led to many extensions dealing with quantitative
aspects, whose transition relations are endowed with further
information like probability rates or stochastic rates; see
\cite{bg98:empa,denicola13:ultras,hillston:pepabook,ks2013:w-s-gsos,pc95:cj} among others.  These calculi are very effective in modelling and
analysing quantitative aspects, like performance analysis of computer
networks, model checking of time-critical systems, simulation of
biological systems, probabilistic analysis of security and safety
properties, etc.

 
Each of these calculi is tailored to a specific quantitative aspect
and for each of them we have to develop a quite complex theory almost
from scratch.  This is a daunting and error-prone task, as it embraces
the definition of syntax, semantics, transition rules,
various behavioural equivalences, logics, proof systems; the proof of
important properties like congruence of behavioural equivalences; the
development of algorithms and tools for simulations, model checking,
etc.
This situation would naturally benefit from general
\emph{frameworks} for LTS with quantitative aspects, i.e.,
mathematical \emph{meta}models offering general methodologies, results, and
tools, which can be uniformly instantiated to a wide range of specific
calculi and models.
In recent years, some of these theories have been proposed; we mention
\emph{Segala systems} \cite{sl:njc95},
\emph{Functional Transition Systems} (FuTS)
\cite{latella:qapl2015},
\emph{weighted labelled transition systems} (WLTSs)
\cite{handbook:weighted2009,ks2013:w-s-gsos},
and \emph{Uniform Labelled Transition Systems} (ULTraS), introduced by
Bernardo, De Nicola and Loreti specifically as ``a uniform setting for
modelling non-deterministic, probabilistic, stochastic or mixed
processes and their behavioural equivalences''
\cite{denicola13:ultras}.

A common feature of most of these meta-models is that their labelled
transition relations do not yield simple states (e.g., processes), but
some mathematical object representing quantitative information about
``how'' each state can be reached.  In particular, transitions in
ULTraS systems have the form $P \xrightarrowu{a} \rho$ where $\rho$ is
a \emph{state reachability weight function}, i.e., a function
assigning a \emph{weight} to each possible state.\footnote{The reader
  aware of advanced process calculi will be not baffled by the fact
  that targets are not processes. Well known previous examples are the
  LTS abstractions/concretions for $\pi$-calculus, for the applied
  $\pi$-calculus, for the ambient calculus, etc.}  By suitably
choosing the set of weights, and how these functions can be combined,
we can recover ordinary non-deterministic LTSs, probabilistic
transition systems, stochastic transition systems, etc.  As
convincingly argued in \cite{denicola13:ultras}, the use of weight
functions in place of plain processes simplifies the combination of
non-determinism with quantitative aspects, like in the case of EMPA or
PEPA.  Moreover, it paves the way for general definitions and results,
an important example being the notion of $\mathcal M$-bisimulation
\cite{denicola13:ultras}.

Albeit quite effective, these meta-models are at their dawn, with many
results and techniques still to be developed.  An important example of
these missing notions is a \emph{specification format}, like the
well-known GSOS, ntyft/ntyxt and ntree formats for non-deterministic
labelled transition systems.  These formats are very useful in
practice, because they can be used for ensuring important properties
of the system; in particular, the bisimulations induced by systems in
these formats is guaranteed to be a congruence (which is crucial for
compositional reasoning).  From a more foundational point of view,
these frameworks would benefit from a categorical characterization in
the theory of coalgebras and bialgebras: this would allow a
cross-fertilizing exchange of definitions, notions and techniques with
similar contexts and theories.


In this paper, we provide two main contributions in this respect.
First, we present a GSOS-style format, called \emph{Weight Function
  GSOS} (WF-GSOS), for the specifications of non-deterministic systems
with quantitative aspects. The judgement derived by rules in this
style is of the form $P \xrightarrowu{a} \psi$, where $P$ is a process
and $\psi$ is a \emph{weight function term}. These terms describe
weight functions by means of an \emph{interpretation}; hence, a
specification given in this format defines a ULTraS. By choosing the
set of weights, the language of weight function terms and their
interpretation, we can readily capture many quantitative notions
(probabilistic, stochastic, etc.), and different kinds of
non-deterministic interactions, covering models like PEPA, TIPP, PCSP,
EMPA, among others.  Moreover, the WF-GSOS format supports a general
definition of \emph{(strong) bisimulation}, which can be readily
instantiated to the various specific systems.

The second contribution is more fundamental.  We provide a general
categorical presentation of these non-deterministic systems with
quantitative aspects. Namely, we prove that ULTraS systems are in
one-to-one correspondence with coalgebras of a precise class of
functors, parametric on the underlying weight structure.  Using this
characterization we define the abstract notion of \emph{WF-GSOS
  distributive law} (i.e.~a natural transformation of a specific
shape) for these functors.  We show that each WF-GSOS specification
yields such a distributive law (i.e., the format is sound); taking
advantage of Turi-Plotkin's bialgebraic framework, this implies that
the bisimulation induced by a WF-GSOS is always a congruence, thus
allowing for compositional reasoning in quantitative settings.
Additionally, we extend the results we presented in \cite{mp:qapl14}
proving that the WF-GSOS format is also \emph{complete}: every
abstract WF-GSOS distributive law for ULTraSs can be described by
means of some WF-GSOS specification.

The rest of the paper is structured as follows.  In
Section~\ref{sec:ultras} we recall Uniform Labelled Transition
Systems, and their bisimulation.  In Section~\ref{sec:WF-GSOS} we
introduce the \emph{Weight Function SOS} specification format for the
syntactic presentation of ULTraSs.
In Section~\ref{sec:examples} we provide some application examples,
such as a WF-GSOS specification for PEPA and the translations of
Segala-GSOS and WGSOS specifications in the WF-GSOS format.  The
categorical presentation of ULTraS and WF-GSOS, with the results that
the format is sound and complete and bisimilarity is a congruence, are
in Section~\ref{sec:coalg}.  Final remarks, comparison with related
work and directions for future work are in Section~\ref{sec:concl}.


\section{Uniform Labelled Transition Systems and their bisimulation}\label{sec:ultras}
In this section we recall and elaborate the definition of ULTraSs, and
define the corresponding notion of (coalgebraically derived)
bisimulation; finally we compare it with the notion of $\mathcal
M$-bisimulation presented in \cite{denicola13:ultras}.
\assumeapx{Additional examples are provided in the Appendix. }%
Although we focus on the ULTraS framework, the results and
methodologies described in this paper can be ported to similar formats
(like FuTS \cite{latella:qapl2015}),
and more generally to a wide range of systems combining 
computational aspects in different ways.

\subsection{Uniform Labelled Transition Systems}
ULTraS are
(non-deterministic) labelled transition systems whose transitions lead
to \emph{state reachability weight functions}, i.e.~functions
representing quantitative information about ``how'' each state can be
reached. Examples of weight functions include probability
distributions, resource consumption levels, or stochastic rates. In
this light, ULTraS can be thought of as a generalization of Segala systems
\cite{sl:njc95}, which stratify non-determinism over probability.
Following the parallel with Segala systems, ULTraS transitions can be
pictured as being composed by two steps:
\[
   x \xrightarrowu{a} \rho \xrightarroww{w} y
\]
where the first is a labelled non-deterministic (sub)transition
and the second is a weighted one; from this perspective the weight function 
plays the r\^ole of the ``hidden intermediate state''.

Akin to Weighted Labelled Transition Systems (WLTS)
\cite{ks2013:w-s-gsos,handbook:weighted2009},
weights are drawn from a fixed
set endowed with a commutative monoid structure, where the unit is
meant to be assigned to disabled transitions (i.e.~those yielding
unreachable states) and the monoidal addition is used to
compositionally weigh sets of transitions given by non-determinism.

\begin{definition}[$\mathfrak W$-ULTraS]
	\label{def:ultras}
	Given a commutative monoid $\mathfrak W = (W,+,0)$, a 
	\emph{($\mathfrak W$-weighted) Uniform Labelled Transition System} 
	is a triple $(X,A,\rightarrowu)$ where:
	\begin{itemize}
		\item 
			$X$ is a set of \emph{states} (processes) called \emph{state space} or \emph{carrier};
		\item 
			$A$ is a set of \emph{labels} (actions);
		\item 
			${\rightarrowu} \subseteq X\times A \times [X \to W]$ 
			is a \emph{transition relation} where $[X \to W]$ 
			denotes the set of all \emph{weight functions} from $X$ to 
			the carrier of $\mathfrak W$.
	\end{itemize}
\end{definition}
Monoidal addition does not play any r\^ole in the above
definition\footnote{%
	Originally, in \cite{denicola13:ultras} $W$ is a partial order
	with bottom.  Actually, the order is not crucial to the basic definition
	of ULTraS as it is only used by some equivalences considered in that
	paper.
} but it is crucial to define the notion of bisimulation
and in general how the ``merging'' of two states (e.g.~induced by functions between carriers) affects the transition relation. In fact, bisimulations
can be thought as inducing ``state space refinements that are well-behaved 
w.r.t.~the transition relation''. From this perspective, monoidal addition
provides an \emph{abstract, uniform and compositional} way to ``merge'' the 
outgoing transitions into one: adding their weight; likewise probabilities
or stochastic rates are added in probabilistic or stochastic systems.

Because the monoidal structure supports finite addition
only\footnote{%
  Indeed it is possible to assume sums for any family indexed by some
  set; however, in Section~\ref{sec:coalg} we assume image-finiteness
  to guarantee the existence of a final coalgebra.  } we can only
merge finitely many transitions.  Assuming ULTraSs to have a finite
carrier or maps between carriers to define finite pre-images
(i.e.~$|f^{-1}(y)| \in \mathbb N$) is preposterous: since we aim to
provide syntactic description of ULTraSs, state spaces may be infinite
(cf.~initial semantics) and functions may map arbitrary many states to
the same image, e.g., their behaviour (cf.~bisimulations, final
semantics).  Therefore, in this paper we shall consider \emph{image
  finite} ULTraSs only.  This is a mild and common assumption
(e.g.~\cite{ks2013:w-s-gsos,bartels04thesis,bloomIM:95}) and our
results readily generalise to transfinite bounds (e.g.~to deal with
countably-branching systems).

\begin{definition}[Image finiteness]
  Let $\mathfrak W = (W,+,0)$ be a commutative monoid.
  For a function $\rho : X \to W$ the set $\support{\rho} \defeq \{ x \mid \rho(x) \neq 0\}$
  is called \emph{support of $\rho$} and whenever it is finite
  $\rho$ is said to be \emph{finitely supported}.
  The set of finitely supported functions with domain $X$
  is denoted by $\fw X$.
  A $\mathfrak W$-ULTraS $(X,A,\rightarrowu)$ is said to be \emph{image finite} 
  iff for any state $x \in X$ and label $a \in A$ the set 
  $\{\rho \mid x\xrightarrowu{a}\rho\}$ is finite and contains
  only finitely supported weight functions.
\end{definition}

\begin{example}
A weight function $\rho \in \ff{\!\!2}X$ 
(for $\mathfrak 2 = (\{\ltrue,\lfalse\},\lor,\lfalse)$)
is a predicate describing a finite subset of $X$.
Thus $\pf X \cong \ff{\!\!2}X$.
Likewise, a function $\rho \in \ff{N}X$ 
(for $\mathfrak N = (\mathbb N,+,0)$) assigns to
each element of $X$ a multiplicity and hence
describes a finite multiset.
\end{example}

Intuitively, elements of $\fw X$ can be seen as ``generalised
multisets''. Therefore, it is natural to extend a function $f : X \to Y$
to a function $\fw(f) : \fw X \to \fw Y$ mapping (finitely supported) 
weight functions over $X$ to (finitely supported) weight functions over 
$Y$ as follows:
\begin{equation}\label{eq:fw-action}
  \textstyle
  \fw(f)(\rho) \defeq \lambda y:Y.
  \sum_{x \in f^{-1}(y)}\rho(x)\text{.}
\end{equation}
This definition generalises the extension of a function to the powerset;
in fact, $\ff{\!\!2}(f)(\rho) = \lambda y : Y.\bigvee_{x \in f^{-1}(y)}\rho(x)$
describes the subset of $Y$ whose elements are image of
some element in the subset of $X$ described by $\rho$.
Henceforth, we shall refer to $\fw(f)(\rho)$ as the 
\emph{action of $f$ on $\rho$} and denote it by $\rho[f]$,
when confusion seems unlikely.

We can now make the idea of 
``state space maps being well-behaved w.r.t.~the transition relation'' 
formal:
\begin{definition}[ULTraS homomorphism]
	Let $(X,A,\rightarrowu_X)$ and $(Y,A,\rightarrowu_Y)$ be two image-finite 
	$\mathfrak W$-ULTraS. A \emph{homomorphism} $f : (\rightarrowu_X) \to (\rightarrowu_Y)$ 
	is a function $f : X \to Y$ between their state spaces 
	such that for any $x \in X$ and $a \in A$:
	\[
		x \xrightarrowu{a}_{\!\!X} \rho \iff f(x) \xrightarrowu{a}_{\!\!Y} \rho[f]
		\text{.}
	\]
\end{definition}
Given two homomorphisms 
$f : (\rightarrowu_X) \to (\rightarrowu_Y)$ and 
$g : (\rightarrowu_Y) \to (\rightarrowu_Z)$,
the function $g \circ f : X \to Z$ is a homomorphism
$g \circ f : (\rightarrowu_X) \to (\rightarrowu_Z)$.
Homomorphism composition is always defined, it is associative
and has identities.
In Section~\ref{sec:coalg} we will show that ULTraSs homomorphisms
indeed form categories equivalent to categories of coalgebras for a suitable
functor. For the time being, consider the degenerate monoid $\mathfrak 1$
containing exactly its unit and let $A$ be a singleton; 
then a $\mathfrak 1$-ULTraS  $(X,A,\rightarrowu_X)$ is just a relation 
$\rightarrowu_X \cong R_X$ on $X$ and
any homomorphism is exactly a relation homomorphism. In fact,
$f : X \to Y$ is a $\mathfrak 1$-ULTraS homomorphism 
$f : (\rightarrowu_X) \to (\rightarrowu_Y)$ iff
$(x,x') \in R_x \iff (f(x),f(x')) \in R_Y$. For $A$ with more
than one label we get exactly homomorphisms of labelled relations
i.e.~LTSs.

\subsection{Bisimulation}
We present now the definition of bisimulation for ULTraS
based on the notion of \emph{kernel bisimulation}
(a.k.a.~behavioural equivalence) i.e.~``a relation which is
the kernel of a common compatible refinement of the two\footnote{%
	We present bisimulations as relations between two
	state spaces instead of considering one system in isolation;
	we are aware that in the case of ULTraS two systems can
	be ``run in parallel'' still the notion of having a common
	refinement allows for different homomorphisms even when
	considering a single system and therefore offers greater
	generality.
} systems''
\cite{staton11}.  This notion naturally stems from the final
semantics approach and, under mild assumptions, coincides
with Aczel-Medler's coalgebraic bisimulation, as we will see in Section~\ref{sec:coalg}.

\begin{definition}[Refinement]
	Given $(X,A,\rightarrowu_X)$ a \emph{refinement} for it
	is any $(Y,A,\rightarrowu_Y)$ such that there exists
	an homomorphism $f : (\rightarrowu_X) \to (\rightarrowu_Y)$.
\end{definition}
Homomorphisms provide the right notion of refinement.
Consider an equivalence relation $R \subseteq X \times X$,
$R$ is \emph{stable w.r.t.~$\rightarrowu_X$}
if, and only if, its equivalence classes are not split by
the transition relation $\rightarrowu_X$, i.e., 
iff there is a refinement whose carrier is $Y = X/R$. 
Hence, stability of an equivalence relation corresponds to the 
canonical projection $\kappa : X \to X/R$ being a ULTraS homomorphism.
This observation contains all the ingredients needed
to define bisimulations for ULTraSs. Before we formalise this 
notion let us introduce some accessory notation.

In the following, we will denote the \emph{total weight} of
$\rho \in \fw X$ by
$\totalweight{\rho} \defeq \sum_{x \in X} \rho(x)$.
The weight $\rho$ assigned to $C \subseteq X$ is
the total weight of the restriction $\restr{\rho}{C}$
i.e.~$\totalweight{\restr{\rho}{C}} = \sum_{x \in C} \rho(x)$.
Any relation $R$ between two sets $X$ and $Y$ defines a relation 
$R_\mathfrak{W}$
between finitely supported weight functions for $X$ and $Y$ as:
\[
	(\phi,\psi) \in R_\mathfrak{W} \defiff
	\forall (C,D) \in R^\star\,
	\totalweight{\restr{\phi}{C}} = 
	\totalweight{\restr{\psi}{D}}
\]
where $R^\star\subseteq \p{}X\times\p{}Y$ is the
\emph{subset closure} of $R$
i.e.~smallest relation s.t., for $C\subseteq X$, $D\subseteq Y$:
\begin{align*}
(C,D)\in R^\star \iff & (\forall x\in C,\forall y\in Y: (x,y)\in R \Rightarrow
y\in D) \wedge  \\& 
 (\forall x\in X,\forall y\in D: (x,y)\in R \Rightarrow x\in C)
\end{align*}

\begin{definition}[Bisimulation]
	\label{def:bisim}
	Let $(X,A,\rightarrowu_X)$ and $(Y,A,\rightarrowu_Y)$ be two image-finite 
	$\mathfrak W$-ULTraS. A relation $R$ between $X$ and $Y$  
	is a \emph{bisimulation} if, and only if,
	for each pair of states $x \in X$ and $y \in Y$, 
	$(x,y) \in R$ implies that for each label $a \in A$ 
	the following hold:
	\begin{itemize}
		\item
		if ${x\xrightarrowu{a}_{\!\!X}\phi}$ then there exists 
	    ${y\xrightarrowu{a}_{\!\!Y}\psi}$ s.t.~$(\phi,\psi) \in R_\mathfrak{W}$.
		\item
		if ${y\xrightarrowu{a}_{\!\!Y}\psi}$ then there exists
	    ${x\xrightarrowu{a}_{\!\!X}\phi}$ s.t.~$(\phi,\psi) \in R_\mathfrak{W}$.
	\end{itemize}
	Processes $x$ and $y$ are said to be \emph{bisimilar} 
	if there exists a bisimulation relation $R$ such that $(x,y) \in R$.
\end{definition}

\looseness=-1
As ULTraSs can be seen as stacking non-determinism over other
computational behaviour, Definition~\ref{def:bisim} stratifies
bisimulation for non-deterministic labelled transition system over
bisimulation for systems expressible as labelled transition systems
weighted over commutative monoids. In fact, two processes $x$ and $y$ are related by some
bisimulation if, and only if, whether one reaches a weight function
via a non-deterministic labelled transition, the other can reach
another function via a transition with the same label, where the
two functions are equivalent in the sense that they assign the
same total weight to the classes of states in the relation.  For
instance, in the case of weights being probabilities, functions are
considered equivalent only when they agree on the probabilities
assigned to each class of states which is precisely the intuition
behind probabilistic bisimulation \cite{ls:probbisim}.
More examples will be discussed below\assumeapx{ and in the Appendix}.

\paragraph{Constrained ULTraS}
Sometimes, the ULTraSs induced by a given monoid are too many, and we
have to restrict to a subclass.  For instance, fully-stochastic
systems such as (labelled) CTMCs are a strict subclass of ULTraSs
weighted over the monoid of non-negative real numbers $(\mathbb
R_0^+,+,0)$, where weights express rates of exponentially distributed
continuous time transitions. In the case of fully-stochastic systems,
for each label, each state is associated with precisely one weight
function. This kind of ``deterministic'' ULTraSs are called
\emph{functional} in \cite{denicola13:ultras}, because the transition
relation is functional, and correspond precisely to
WLTSs~\cite{ks2013:w-s-gsos,handbook:weighted2009}.
These are a well-known family of systems (especially their automata
counterpart) and have an established coalgebraic understanding as long
as a (coalgebraically derived) notion of \emph{weighted bisimulation}
which are shown to subsume several known kinds of systems such as
non-deterministic, (fully) stochastic, generative and reactive
probabilistic \cite{ks2013:w-s-gsos}. Moreover,
Definition~\ref{def:bisim} coincides with weighted bisimulation on
functional ULTraSs/WLTSs over the same monoid
\cite[Def.~4]{ks2013:w-s-gsos}; hence Definition~\ref{def:bisim}
covers every system expressible in the framework of WLTS.
\assumeapx{(cf.~Appendix~\ref{apx:vs-wlts}).}
\begin{proposition}
  \label{prop:w-bisim}
  Let $\mathfrak W$ be a commutative monoid and $(X,A,\rightarrowu_X)$,
  $(Y,A,\rightarrowu_Y)$
  be $\mathfrak W$-LTSs seen as a functional $\mathfrak W$-ULTraSs.
  Every bisimulation relation between them is a $\mathfrak W$-weighted bisimulation
  and vice versa.
\end{proposition}
\ifappendix
\begin{proof}[Proof (Omitted)]
See Appendix~\hyperref[proof:w-bisim]{\ref*{apx:vs-wlts}}.
\end{proof}
\fi

Another constraint arises in the case of probabilistic
systems, i.e., weight functions are probability distributions.  Since
addition is not a closed operation in the unit interval $[0,1]$, there
is no monoid $\mathfrak W$ such that every weight function on it is also a
probability distribution.  Altough we could relax
Definition~\ref{def:ultras}
to allow commutative \emph{partial} monoids\footnote{A commutative
  partial monoid is a set endowed with a unit and a partial binary
  operation which is associative and commutative, where it is defined,
  and always defined on its unit. } such as the weight structure of
probabilities $([0,1],+,0)$, not every weight function
on $[0,1]$ is a probability distribution.  In fact, probabilistic
systems (among others) can be recovered as ULTraSs over the $(\mathbb
R^+_0,+,0)$ (i.e.~the free completion of $([0,1],+,0)$) and subject to
suitable constraints.  For instance, Segala systems \cite{sl:njc95}
are precisely the strict subclass of $\mathbb R^+_0$-ULTraS such that
every weight function $\rho$ in their transition relation is a
probability distribution i.e.~$\totalweight\rho = 1$.  Moreover,
bisimulation is preserved by constraints; e.g., bisimulations
on the above class of (constrained) ULTraS corresponds to Segala's (strong)
bisimulations \cite[Def.~13]{sl:njc95}.
\begin{proposition}\label{prop:segala-bisim}
  Let  $(X,A,\rightarrowu_X)$ and $(Y,A,\rightarrowu_Y)$
  be image-finite Segala-systems seen as
  ULTraSs on $(\mathbb R^+_0,+,0)$.
  Every bisimulation relation between them is a strong bisimulation
  in the sense of \cite[Def.~13]{sl:njc95} and vice versa.
\end{proposition}
\ifappendix
\begin{proof}[Proof (Omitted)]
See Appendix~\hyperref[proof:segala-bisim]{\ref*{apx:vs-segala}}.
\end{proof}
\fi

A similar result holds for generative (or fully) or reactive
probabilistic systems and their bisimulations.
In fact, these are functional ${\mathbb R^0_+}$-ULTraS 
s.t.~for all $x \in X$ 
${x \xrightarrowu{a} \rho} \implies \totalweight{\rho} \in \{0,1\}$
and
$\sum_{\{\rho \mid x \xrightarrowu{a} \rho\}}\totalweight{\rho} \in \{0,1\}$
respectively.

\subsection{Comparison with $\mathcal M$-bisimulation}
Bernardo et al.~defined a notion of bisimulation for ULTraS
parametrized by a function $\cal M$ which is used to weight sets of
(sequences of) transitions \cite[Def.~3.3]{denicola13:ultras}. 
Notably, $\cal M$'s
codomain may be not the same of that used for weight functions in the
transition relation. This offers an extra degree of freedom with
respect to Definition~\ref{def:bisim}. 
We recall the relevant definitions with minor modifications since
the original ones have to consistently weight also sequences of transitions
in order to account also for trace equivalences which are not in the scope of this paper.

\begin{definition}[$M$-function]
  \label{def:mfun}
  Let $(M,\bot)$ be a pointed\footnote{%
	  A pointed set (sometimes called based set or rooted set)
	  is a set equipped with a distinguished element
	  called (base) point; homomorphisms are
	  point preserving functions.
  } set and $(X,A,\rightarrowu)$ be a 
  $\mathfrak W$-ULTraS.
  A function $\mathcal M : X \times A\times \mathcal P X \to M$ is an
  \emph{$M$-function for $(X,A,\rightarrowu)$} if, and only if,
  it agrees with termination and class union, i.e.:
	\begin{itemize}
		\item
			for all $x \in X$, $a\in A$ and $C \in \mathcal PX$, 
			$\mathcal M(x,a,C) = \bot$ whenever 
			${x \centernot{\xrightarrowu{a}}}$ or
			$\totalweight{\restr{\rho}{C}} = 0$ for every $x \xrightarrowu{a}\rho$;
		\item
			for all $x,y \in X$, $a \in A$ and  $C_1,C_2 \in \mathcal P X$, if
			$\mathcal{M}(x,a,C_1) = \mathcal{M}(y,a,C_1)$ and $\mathcal{M}(x,a,C_2) = \mathcal{M}(y,a,C_2)$ then
			$\mathcal{M}(x,a,C_1\cup C_2) = \mathcal{M}(y,a,C_1 \cup C_2)$.
	\end{itemize}
\end{definition}

\begin{definition}[$\mathcal M$-bisimulation {\cite{denicola13:ultras}}]
  \label{def:mbisim}
  Let $\mathcal M$ be an $M$-function for $(X,A,\rightarrowu)$.
  An equivalence relation $R \subseteq X\times X$ is a \emph{$\cal M$-bisimulation for $\rightarrowu$} 
  iff for each pair $(x,y)\in R$, label $a \in A$, and class $C \in X/R$,
  $\mathcal{M}(x,a,C) = \mathcal{M}(y,a,C)$.
\end{definition}

\looseness=-1
Differently from Definition~\ref{def:bisim}, $M$ may be not $W$
allowing one to, for instance, consider stochastic rates up-to a 
suitable tolerance as a way to account for experimental measurement 
errors in the model.
A further distinction between bisimulation and $\mathcal M$-bisimulation
arises from the fact that ULTraSs come with two distinct ways of
\emph{terminating}.  A state can be seen as ``terminated'' either when
its outgoing transitions are always the constantly zero function, or
when it has no transitions at all.  In the first case, the state has
still associated an outcome, saying that no further state is
reachable; we call these states \emph{terminal}. In the second case,
the LTS does not even tell us that the state cannot reach any further
state; in fact, there is no ``meaning'' associated to the state. In
this case, we say that the state is \emph{stuck}.\footnote{This is
  akin to sequential programs: a terminal state is when we reach the
  end of the program; a stuck state is when we are executing an
  instruction whose meaning is undefined.} The bisimulation given in
Definition~\ref{def:bisim} keeps these two terminations as different
(i.e., they are not bisimilar), whereas $\mathcal M$-bisimulation does not
make this distinction (cf.~\cite[Def.~3.2]{denicola13:ultras} or, for
a concrete example based on Segala systems,
\cite[Def.~7.2]{denicola13:ultras}).

Finally, the two notions differ on the quantification over equivalence
classes: in the case of Definition~\ref{def:bisim} quantification
depends on the non-deterministic step whereas in the case of $\mathcal
M$-bisimulation it does not.

Under some mild assumptions, the two notions agree.  In
particular, let us restrict to systems with just one of the two
terminations for each action $a$---i.e.~if for some $x$, $\{\rho\mid
x\xrightarrowu{a}\rho\} = \emptyset$ then for all $y$, $\lambda z. 0
\notin \{\rho\mid y\xrightarrowu{a}\rho\}$, and, symmetrically, if for
some $x$, $\lambda z. 0 \in \{\rho\mid x\xrightarrowu{a}\rho\}$ then
for all $y$, $\{\rho\mid y\xrightarrowu{a}\rho\} \neq \emptyset$.
Then, the bisimulation given in Definition~\ref{def:bisim} corresponds
to a $\mathcal M$-bisimulation for a suitable choice of $\mathcal M$.
\begin{proposition}\label{prop:vs-m-bisim}
	Let $(X,A,\rightarrowu)$ be a $\mathfrak W$-ULTraS with at most one kind
	of termination, for each label.
	Every bisimulation $R$ is also an $\mathcal M$-bisimulation for
	\[
		\mathcal{M}(x,a,C) \defeq 
		\{[\rho]_{R_{\mathfrak{W}}} \mid {x\xrightarrowu{a}\rho}
		\text{ and } 
		\totalweight{\restr{\rho}{C}} \neq 0\}
		\cup
		\{[\lambda z. 0]_{R_{\mathfrak{W}}}\}
	\]
	where 
	$(M,\bottom) = (\pf(\fw X/{R_{\mathfrak{W}}}), \{[\lambda z. 0]_{R_{\mathfrak{W}}}\})$.
\end{proposition}
\begin{proofatend}
The function $\mathcal M$ is well-given because 
$\mathcal M(x,a,C) = \bot = \bottom$ whenever 
$x \centernot{\xrightarrowu{a}}$ or, for each
$x \xrightarrowu{a} \rho$, $\rho(C) = 0$, and
$\mathcal{M}(x,a,C_1) = \mathcal{M}(y,a,C_1)$ and $\mathcal{M}(x,a,C_2) = \mathcal{M}(y,a,C_2)$
implies $\mathcal{M}(x,a,C_1\cup C_2) = \mathcal{M}(y,a,C_1 \cup C_2)$
by definition of $R_{\mathfrak{W}}$.

By Definition~\ref{def:bisim}, whenever $x \xrightarrowu{a} \phi$
then $y \xrightarrowu{a} \psi$ s.t.~$\phi(C) = \psi(C)$
for each $C\in X/R$ i.e.~$\phi R_{\mathfrak{W}} \psi$ and the symmetric case for $y$.
Therefore $(x,y) \in R$ implies that 
$\Phi_{x,a} \defeq \{[\phi]_{R_{\mathfrak{W}}} \mid x \xrightarrowu{a} \phi\}$
and
$\Phi_{y,a} \defeq \{[\psi]_{R_{\mathfrak{W}}} \mid y \xrightarrowu{a} \psi\}$ 
are equal for each $a \in A$.
We can safely add $\bottom$ to both $\Phi_{x,a}$
and $\Phi_{y,a}$ since, whenever both $x$ and $y$ terminate,
they are either both stuck or both terminal.
In fact, equality and inequality are preserved
while adding $\bottom$ since $\Phi_{x,a} = \emptyset \implies
\bottom \notin \Phi_{y,a}$ (and vice versa) by hypothesis.
For each $C\in X/R$ ($x,y \in X$ and $a \in A$) let $\Psi_{x,a,C} \defeq 
(\Phi_{x,a}\setminus\{[\rho]_{R_{\mathfrak{W}}}\mid\rho(C) = 0\})\cup\{\bottom\}$.
Clearly $\Phi_{x,a} \cup = \bigcup_{C \in X/R} \Psi_{x,a,C}$
and if $(x,y) \in R$ then $\Psi_{x,a,C} = \Psi_{y,a,C}$. 
Complementarly, if $(x,y) \notin R$
then there exists some $\phi \in \Phi_{x,a}$ s.t.~for no $\psi \in \Phi_{y,a}$
$\phi R_{\mathfrak{W}} \psi$ or vice versa; w.l.o.g.~assume the former.
Hence there exists $C \in X/R$ such that $\phi(C) \neq \psi(C)$
whence $\Psi_{x,a,C} \neq \Psi_{y,a,C}$. Finally, we conclude by
$\mathcal M(x,a,C) = \Psi_{x,a,C}$ for each $x \in X$, $a\in A$ and $C \in X/R$.
\end{proofatend}

Intuitively, Definition~\ref{def:bisim} generalises strong bisimulation
for Segala systems (Segala and Lynch's probabilistic bisimilarity \cite{sl:njc95}) 
and $\mathcal{M}$-bisimulation generalises convex bisimulation
\cite{denicola13:ultras}.

\section{WF-GSOS: A complete GSOS format for ULTraSs}\label{sec:WF-GSOS}
In this section we introduce the \emph{Weight Function SOS} 
specification format for the syntactic presentation of ULTraSs.  As it
will be proven in Section~\ref{sec:cong-proof}, bisimilarity for
systems given in this format is guaranteed to be a congruence with
respect to the signature used for representing processes.

The format is parametric in the weight monoid $\mathfrak W$ and, as usual,
in the \emph{process signature} $\Sigma$ defining the syntax of system
processes.  In contrast with ``classic'' GSOS formats
\cite{klin:tcs11}, targets of rules are not processes but terms whose
syntax is given by a different signature, called the \emph{weight
  signature}.  This syntax can be thought of as an ``intermediate
language'' for representing weight functions along the line of viewing
ULTraSs as stratified (or staged) systems. An early example of this
approach can be found in \cite{bm:2015stocsos}, where targets are terms
representing measures over the continuous state space.
Earlier steps in this direction can be found e.g.~in Bartels' GSOS
format for Segala systems (cf.~\cite[§5.3]{bartels04thesis} and \cite[§4.2]{mp:qapl14}) or in
\cite{cm:quest10,denicola13:ultras} where targets are described by
meta-expressions.

\begin{definition}[WF-GSOS Rule]\label{def:wf-gsos-rule}
	Let $\mathfrak W$ be a commutative monoid and $A$ a set of labels. 
	Let $\Sigma$ and $\Theta$ be the \emph{process signature} and 
	the \emph{weight signature}, respectively.
	A WF-GSOS rule over them is a rule of the form: 
	\[\frac{
		\begin{array}{c}
			\Big\{
			x_i \xrightarrowu{a} \phi^a_{ij}
			\Big\}
			\hspace{-1.2ex}\begin{array}{l}
				\scriptstyle  1 \leq i \leq n,\\[-4pt]
				\scriptstyle  a \in A_i,\\[-4pt]
				\scriptstyle  1\leq j \leq m^a_i
			\end{array}
			\quad
			\Big\{
			x_i \centernot{\xrightarrowu{b}}
			\Big\}
			\hspace{-1.2ex}\begin{array}{l}
				\scriptstyle  \\[-4pt]
				\scriptstyle  1 \leq i \leq n,\\[-4pt]
				\scriptstyle  b \in B_i
			\end{array}
			\quad
			\Big\{
			\totalweight{\phi^{a_k}_{i_kj_k}} = w_k
			\Big\}
			_{1 \leq k \leq p}
			\quad
			\Big\{
			\corestr{\phi^{a_k}_{i_kj_k}}{\mathfrak C_k} \ni y_k
			\Big\}
			_{1 \leq k \leq q}
			\end{array}
		}{
			\mathtt{f}(x_1,\dots,x_n) \xrightarrowu{c} \psi
		}\]
	where:
	\begin{itemize}\itemsep=0pt
		\item 
			$\mathtt f$ is an $n$-ary symbol from $\Sigma$;
		\item 
			$X = \{x_i\mid 1 \leq i \leq n\}$, 
			$Y = \{y_k\mid 1 \leq k \leq q\}$ 
			are sets of pairwise distinct \emph{process} variables;
		\item 
			$\Phi = \{\phi^a_{ij}\mid 1 \leq i \leq n,\ a \in A_i,\ 1\leq j \leq m^a_i\}$
			is a set of pairwise distinct  \emph{weight function} variables;
		\item
			$\{w_k \in \mathfrak W \mid 1 \leq k \leq p\}$ are 
			\emph{weight constants};
		\item
			$\{\mathfrak C_k \mid 1 \leq k \leq q,\, w_k \in \mathfrak C_k\}$ is a set of \emph{clubs} of $\mathfrak W$,
			i.e.~subsets of $W$ being monoid ideals whose complements are sub-monoids of $\mathfrak W$;
		\item 
			$a,b,c \in A$ are labels and $A_i \cap B_i = \emptyset$ for  
			$1 \leq i \leq n$;
		\item
			$\psi$ is a \emph{weight term} for the signature $\Theta$ 
			such that $\mathrm{var}(\psi) \subseteq X \cup Y \cup \Phi$.
	\end{itemize}
	A rule like above is \emph{triggered} by 
		a tuple $\langle C_1,\dots,C_n\rangle$ of \emph{enabled labels} and 
		by a tuple $\langle v_1,\dots,v_p \rangle$ of weights 
	if, and only if, 
		$A_i \subseteq C_i$, $B_i \cap C_i = \emptyset$, and
		$w_j = v_j$
	for 
		$1 \leq i \leq n$ and
		$1 \leq j \leq p$.
\end{definition}
Intuitively, the four families of premises can be grouped in two
kinds: the first two families correspond to the non-deterministic (and
labelled) behaviour, whereas the other two correspond to the weighting
behaviour of quantitative aspects.  The former are precisely the
premises of GSOS rules for LTSs (up-to targets being functions), and
describe the possibility to perform some labelled transitions.  The
latter are inspired by Bartels' \emph{Segala-GSOS}
\cite[§5.3]{bartels04thesis} and Klin's WGSOS \cite{ks2013:w-s-gsos} formats;
a premise like $\totalweight{\phi} = w$ constrains the variable
$\phi$ to those functions whose total weight is exactly the constant $w$;
a premise like $\corestr{\phi}{\mathfrak C}\ni y$ binds the process variable
$y$ to those elements being assigned a weight in $\mathfrak C$.
This kind of premises are meant to single out elements from weight 
functions domain in a way that is coherent w.r.t.~function actions (hence
independent from carrier maps and variable substitutions). 
To this end, selection may depend on weights only and has to be 
unaffected by sums, i.e., $z = f(x) = f(x')$ is selected if and only if at 
least $x$ or $x'$ is.
Clubs are the finest substructures of commutative monoids that are 
``isolated'' w.r.t. the monoidal operation in the sense that:
\begin{itemize}
	\item
		are commutative monoid ideals, i.e.~subsets $\mathfrak C$ with a module structure;
	\item
		their complement $\overline{\mathfrak C}$ in $\mathfrak W$ is a sub-monoid of $\mathfrak W$.
\end{itemize}
Because of the first assumption 
$v+w \in \mathfrak C \implies v \in \mathfrak C \lor w \in \mathfrak C$
and because of the second
$v,w \notin \mathfrak C \implies v+w \notin \mathfrak C$
In other words, if something is selected depending on its weight, no matter what is added to, it will remain selected and vice versa:
$v+w \in \mathfrak C \iff v \in \mathfrak C \lor w \in \mathfrak C$.
Note that no club can contain the unit $0$ (otherwise $\overline{\mathfrak C} = \emptyset$) and this ensures selections
to be confined within the weight function supports (hence to be finite).
\begin{remark}
	The empty set trivially is a club.
	Not all complements of submonoids are clubs, for instance 
	even natural numbers under addition are a submonoid of 
	$(\mathbb{N},+,0)$ but odd numbers are not a club;
	the only non-empty club in $(\mathbb{N},+,0)$ is $\mathbb{N}\setminus\{0\}$.
	Elements with an opposite cannot be part of a club: if
	$x \in \mathfrak C$ then $x+(-x) = 0$ is in $\mathfrak C$
	and hence $\overline{\mathfrak C}$ cannot be a submonoid of $\mathfrak W$.
\end{remark}

Like Segala-GSOS (but unlike WGSOS), there are no
variables denoting the weight of each $y_k$ since this information can be
readily extracted from $\phi^{a_k}_{i_k j_k}$, e.g.~by some
operator from $\Theta$ that ``evaluates'' $\phi^{a_k}_{i_k j_k}$
on $y_k$.
Targets of transitions defined by these rules are terms generated from
the signature $\Theta$.  In order to characterize transition relations
for ULTraSs, we need to \emph{evaluate} these terms to weight
functions.  This is obtained by adding an \emph{interpretation for
weight terms}, besides a set of rules in the above format.

Before defining interpretations and specifications, we need to
introduce some notation.  For a signature $S$ and a set $X$ of
variable symbols, let $T^S X$ denote the set of terms freely
generated by $S$ over the variables $X$ (in the
following, $S$ will be either $\Sigma$ or $\Theta$).  A substitution
for symbols in $X$ is any function $\sigma:X \to Y$;
its action extends to terms defining the function $T^S(\sigma) :
T^S X \to T^S Y$ (i.e.~simultaneous substitution). When confusion
seems unlikely we use the more evocative notation $\mathtt t[\sigma]$ instead of
$T^S(\sigma)(\mathtt t)$.

\begin{definition}[Interpretation]\label{def:wf-gsos-eval}
Let $\mathfrak W$ be a commutative monoid, $\Sigma$ and $\Theta$ be 
the process and the weight signature respectively. A
\emph{weight term interpretation}
for them is a family of functions 
\[
  \llbrace\hole\rrbrace_X : T^\Theta(X + \fw(X)) \to \pf\fw
  T^\Sigma(X)
\]
indexed over sets of variable symbols, and respecting substitutions, i.e.:
\[
	\forall \sigma : X \to Y, \psi\in T^\Theta(X):
	\llbrace\psi\rrbrace_X[\sigma] = 
	\llbrace \psi[\sigma] \rrbrace_Y
\text{.}\]
\end{definition}
Different from \cite{mp:qapl14} interpretations allow one term to
represent finitely many weight functions. This generalization offers 
more freedom in the use of the format by reducing the
constrains on what can be encoded in weight function terms
and simplifies the proof for completeness.

We are ready to introduce the WF-GSOS specification format. 
Basically, this is a set of WF-GSOS rules,
subject to some finiteness conditions to ensure image-finiteness,
together with an interpretation. 
\begin{definition}[WF-GSOS specification]\label{def:wf-gsos-spec}
  Let $\mathfrak W$ be a commutative monoid, $A$ the set of labels, $\Sigma$
  and $\Theta$ the process and the weight signature respectively.
  An \emph{image-finite WF-GSOS specification over 
  $\mathfrak W, A, \Sigma$ and $\Theta$} is a pair $\langle\mathcal R,
  \llbrace\hole\rrbrace\rangle$ where
  $\llbrace\hole\rrbrace$ is a weight term interpretation
  and $\mathcal R$ is a set of rules compliant with
  Definition~\ref{def:wf-gsos-rule} and such that only finitely many
  rules share the same operator in the source ($\mathtt f$), the same
  label in the conclusion ($c$), and the same trigger $\langle
  A_1,\dots,A_n\rangle$, $\langle w_1,\dots,w_p\rangle$.
\end{definition}

Every WF-GSOS specification induces an 
ULTraS over ground process terms.
\begin{definition}[Induced ULTraS]\label{def:induced-ultras}
	The ULTraS induced by an image-finite WF-GSOS specification
	$\langle\mathcal R,\llbrace\hole\rrbrace\rangle$ over $\mathfrak
	W, \Sigma, \Theta$ is the $\mathfrak W$-ULTraS $(T^\Sigma\emptyset, A,
	\rightarrowu)$ where $\rightarrowu$ is defined as the smallest subset
	of $T^\Sigma\emptyset
	\times A \times \fw T^\Sigma\emptyset$ being closed under the following condition.

	Let $p = \mathtt{f}(p_1, \dots, p_n) \in T^\Sigma\emptyset$. Since the ground 
	$\Sigma$-terms $p_i$ are structurally smaller than $p$ assume 
	   (by structural recursion)
	that the set $\{\rho \mid p_i \xrightarrowu{a} \rho\}$ --
	and hence the trigger $\vec{A} = \langle A_1,\dots,A_n\rangle$, 
	$\vec{w} = \langle w_1,\dots,w_q\rangle$ --
	is determined for every $i \in \{1,\dots,n\}$ and $a\in A$.
	For any rule $R \in \mathcal R$ whose conclusion is of the form
	$\mathtt{f}(x_1,\dots,x_n) \xrightarrowu{c} \psi$ and triggered by 
	$\vec A$ and $\vec w$
	let $X$, $Y$, $\Phi$ be the set of process and weight function
	variables involved in $R$ as per Definition~\ref{def:wf-gsos-rule}.
	Then, for any substitution $\sigma:X\cup Y  \to  T^\Sigma\emptyset$
	and map $\theta : \Phi \to \fw T^\Sigma\emptyset$ such that:
	\begin{enumerate}
		\item
			$\sigma(x_i) = p_i$ for $x_i \in X$;
		\item
		$\theta(\phi^a_{ij}) = \rho$ 
			for each premise $x_i\xrightarrowu{a}\phi^a_{ij}$ 
			and $\totalweight{\phi^{a}_{ij}} = w_k$ of $R$, 
			and for any $\rho$ such that $p_i\xrightarrowu{a} \rho$ 
			and $\totalweight{\rho} = w_k$;
		\item
			$\sigma(y_k) = q_k$ for each premise 
			$\corestr{\phi^{a_k}_{i_k j_k}}{\mathfrak C_k}\ni y_k$ 
			of $R$ and for any $q_k\in T^\Sigma\emptyset$ s.t.~			$\theta(\phi^{a_k}_{i_k j_k})(q_k) \in \mathfrak C_k$;
	\end{enumerate}
	there is $p \xrightarrowu{c} \rho$ where 
	$\rho \in \llbrace \psi[\theta]\rrbrace_{X\cup Y}[\sigma]$
	is an instantiated interpretation of the target $\Theta$-term $\psi$.
\end{definition}

The above definition is well-defined since it is based on structural recursion
over ground $\Sigma$-terms (i.e.~the process $p$ in each triple $(p,a,\rho)$);
in particular, terms have finite depth and only structurally smaller terms
are used by the recursion (i.e.~the assumption of $p_i \xrightarrowu{a} \rho$
being defined for each $p_i$ in $p = \mathtt{f}(p_1, \dots, p_n)$).
Moreover, for any trigger, operator, and conclusion label 
only finitely many rules have to be considered.

Finally we can state the main result for the proposed format.
\begin{theorem}[Congruence]\label{th:congruence-1}
	The bisimulation on the ULTraS induced by a
	WF-GSOS specification is a congruence with respect 
	to the process signature.
\end{theorem}
The proof is postponed to Section~\ref{sec:cong-proof},
where we will take advantage of the bialgebraic framework.

\begin{remark}[Expressing interpretations]\rm
  Weight term interpretation can be defined in many
  ways, e.g.~by structural recursion on $\Theta$-terms.
  For instance, every substitution-respecting family of maps:
\[
  h_X : \Theta\fw T^\Sigma(X) \to \pf\fw T^\Sigma(X)
  \qquad 
  b_X : X \to \pf\fw T^\Sigma(X)
\]
uniquely extends to an interpretation by structural recursion on
$\Theta$-terms where $h_X$ and $b_X$ define the inductive and base
cases respectively. These maps can be easily given by means of a set
of equations, as in~\cite[§4.1]{mp:qapl14}.
\end{remark}

\section{Examples and applications of WF-GSOS specifications}
\label{sec:examples}
In this section we provide some examples of applications of the WF-GSOS format.
First, we show how a process calculus can be given a WF-GSOS
specification; in particular, we consider PEPA, a well known process
algebra with quantitative features.  Then we show that Klin's Weighted
GSOS format for weighted systems \cite{ks2013:w-s-gsos} and Bartels'
Segala-GSOS format for Segala systems \cite{bartels04thesis} are
subsumed by our WF-GSOS format; this corresponds to the fact that
ULTraSs subsume both weighted and Segala systems.

\subsection{WF-GSOS for PEPA}\label{sec:pepa-WF-GSOS}
In PEPA \cite{hillston:pepabook,hillston05}, processes are terms over the grammar:
\begin{equation}
  \label{eq:pepa-grammar}
  P ::= (a,r).P \mid P + P \mid P \pepasync{L} P 
  \mid P\setminus L 
\end{equation}
where $a$ ranges over a fixed set of labels $A$, $L$ over subsets of $A$ and
$r$ over $\mathbb R^+$.
The semantics of process terms is usually defined by the inference rules
in Figure~\ref{fig:pepa-classic-sos}
\begin{figure}
\[\begin{array}{c}
  \frac{}{(a,r).P \xrightarrowg{a,r} P}
\quad
  \frac{P_1 \xrightarrowg{a,r} Q}{P_1 + P_2 \xrightarrowg{a,r} Q}
\quad
  \frac{P_2 \xrightarrowg{a,r} Q}{P_1 + P_2 \xrightarrowg{a,r} Q}
\quad
  \frac{P \xrightarrowg{a,r} Q}{P \setminus L \xrightarrowg{a,r} Q
  }\  a \notin L
\quad  
  \frac{P \xrightarrowg{a,r} Q}{
    P \setminus L \xrightarrowg{\tau,r} Q
  }\ a \in L
\\
  \frac{
    P_1 \xrightarrowg{a,r_1}Q_1 \quad P_2 \xrightarrowg{a,r_2}Q_2}{
    P_1 \pepasync{L}P_2 \xrightarrowg{a,R} Q_1 \pepasync{L} Q_2
  }\ a \in L
\quad
  \frac{P_1 \xrightarrowg{a,r}Q}{
    P_1 \pepasync{L}P_2 \xrightarrowg{a,r} Q \pepasync{L} P_2
  }\ a \notin L
\quad
  \frac{P_2 \xrightarrowg{a,r}Q}{
    P_1 \pepasync{L}P_2 \xrightarrowg{a,r} P_1 \pepasync{L} Q
  }\ a \notin L
\end{array}\]\vspace{-1ex}
\caption{Structural operational semantics for PEPA.}
\label{fig:pepa-classic-sos}
\end{figure}
where $a \in A$, $r,r_1,r_2,R \in \mathbb R^+$ 
(passive rates are omitted for simplicity) and $R$ depends only on
$r_1$, $r_2$ and the intended meaning of synchronisation.
For instance, in applications to performance evaluation 
\cite{hillston:pepabook}, rates model time and $R$ is 
defined by the \emph{minimal rate law}:
\begin{equation}
\label{eq:minimal-rate-law}
  R = \frac{r_1}{r_a(P_1)}\cdot
  \frac{r_2}{r_a(P_2)}\cdot
  \min(r_a(P_1),r_a(P_2))
\end{equation}
where $r_a$ denotes the apparent rate of $a$ \cite{hillston:pepabook}.

PEPA can be characterized by a specification in the WF-GSOS format where
the process signature $\Sigma$ is the same as \eqref{eq:pepa-grammar}
and weights are drawn from the monoid of positive real numbers under
addition extended with the $+\infty$ element (only for technical
reasons connected with the ${\llbrace\hole\rrbrace}$ and
process variables---differently from other stochastic process algebras
like EMPA \cite{bg98:empa}, PEPA does not allow instantaneous actions,
i.e.~with rate $+\infty$).  The intermediate language of weight terms
is expressed by the grammar:
\[
\theta ::= \bottom \mid \diamondsuit_r(\theta) \mid \theta_1\oplus \theta_2 \mid
\theta_1\parallel_L \theta_2 \mid \xi \mid P
\]
where $r\in \mathbb R^+_0$, $L \subseteq A$, $\xi$ range over weight
functions $\fw X$, and $P$ over processes in $T^\Sigma X$ for some set $X$. 
Note that the grammar is untyped
since it describes the terms freely generated by the signature $\Theta
= \{\bottom:0,\diamondsuit_r:1,\oplus:2,\parallel_L:2\}$, over
weight function variables and processes.  Intuitively $\bottom$ is
the constantly $0$ function, $\diamondsuit_r$ reshapes its argument to
have total weight $r$, $\oplus$ is the point-wise sum and
$\parallel_L$ parallel composition e.g.~by
\eqref{eq:minimal-rate-law}. The formal meaning of these operators is
given below by the definition (by structural recursion on
$\Theta$-terms) of the interpretation
${\llbrace\hole\rrbrace}$ which is introduced alongside
WF-GSOS rules for presentation convenience.  Each operator is
interpreted as a singleton (PEPA describes functional ULTraSs) and
hence we will describe ${\llbrace\hole\rrbrace}$ as if a
weight function is returned.

For each action $a\in A$ and rate $r \in \mathbb R^+$, a process
$(a,r).P$ presents exactly one $a$-labelled transition ending in the
weight function assigning $r$ to the (sub)process denoted by the variable $P$ 
and $0$ to everything else. Hence, the \emph{action axiom} is expressed as
follows:
\[
  \frac{}{(a,r).P \xrightarrowu{a} \diamondsuit_r(P)}
  \qquad
  \llbrace\diamondsuit_r(\psi)\rrbrace_{X}(t) = 
    \begin{cases}
      \frac{r}{
        \raisebox{-1pt}{$\scriptstyle\left|
          \raisebox{-1pt}{$\scriptstyle
          \support{\llbrace\psi\rrbrace_{X}\!}
         $}\right|$}} & \text{if } \llbrace\psi\rrbrace_{X}(t)\neq 0\\
      0 & \text{otherwise}
    \end{cases}
\]
where $\diamondsuit_r$ normalises\footnote{Since the interpretation
  $\llbrace\hole\rrbrace$ is being defined by structural
  recursion and has to cover all the language freely generated from
  $\Theta$, we can not use the (slightly more intuitive) ``Dirac''
  operator $\delta_r(P)$ where $P$ is restricted to be a process
  variable instead of a $\Theta$-term. Likewise, indexing
  $\delta_{r,P}$ also over processes would break substitution
  independence i.e.~naturality.} $\llbrace P \rrbrace_{X}$
to equally distribute the weight $r$ over its support; in particular,
since process variables will be interpreted as ``Dirac-like''
functions $\diamondsuit_r(P)$ corresponds to the weight function
assigning $r$ to $\Sigma$-term denoted by $P$.

Conversely to the action axiom, $(a,r).P$ can not perform any action but $a$:
\[
  \frac{}{(a,r).P \xrightarrowu{b} \bottom}\ a \neq b
  \qquad
  \llbrace \bottom \rrbrace_{X}(t) = 0
\]
This rule is required to obtain a functional ULTraS and 
is implicit in Figure~\ref{fig:pepa-classic-sos}
where disabled transitions are assumed with rate $0$
as in any specification in the Stochastic GSOS or Weighted 
GSOS formats. Without this rule, transitions would have 
been disabled in the non-deterministic layer 
i.e.~$(a,r).P\centernot{\xrightarrowu{b}}$.

Stochastic choice is resolved by the stochastic race, hence the rate
of each competing transition is added point-wise as in
Figure~\ref{fig:pepa-classic-sos} (and in the SGSOS and WGSOS
formats).  This passage belongs to the stochastic layer of the
behaviour (hence to the interpretation, in our setting) whereas the
selection of which weight functions to combine is in the
non-deterministic behaviour represented by the rules and, in
particular, to the labelling. Therefore, the \emph{choice rules}
become:
\[
  \frac{P_1 \xrightarrowu{a} \phi_1 \quad P_2 \xrightarrowu{a} \phi_2}{
    P_1 + P_2 \xrightarrowu{a} \phi_1 \oplus \phi_2
  }\qquad
  \llbrace \psi \oplus \phi \rrbrace_{X}(t) = 
    \llbrace \psi \rrbrace_{X}(t) + 
    \llbrace \phi \rrbrace_{X}(t)
\]
Likewise, process cooperation depends on the labels to select the
weight function to be combined. This is done in the next two rules:
one when the two processes cooperate, and the other when one process
does not interact on the channel:
\begin{gather*}
 \frac{P_1 \xrightarrowu{a} \phi_1 \quad P_2 \xrightarrowu{a} \phi_2}{
    P_1 \pepasync{L} P_2 \xrightarrowu{a} \phi_1 \parallel_{L} \phi_2
  } a \in L
\quad
  \frac{P_1 \xrightarrowu{a} \phi_1 \quad P_2 \xrightarrowu{a} \phi_2}{
    P_1 \pepasync{L}P_2 \xrightarrowu{a}  
    (\phi_1 \parallel_{L} P_2) \oplus
    (P_1 \parallel_{L} \phi_2)
  } a \notin L
\end{gather*}
The combination step depends on the minimal rate law \eqref{eq:minimal-rate-law}:
\[
\llbrace \psi \parallel_{L} \phi \rrbrace_{X}(t) = 
\begin{cases}
  \frac{\llbrace \psi \rrbrace_{X}(t_1)}{
    \totalweight{\llbrace \psi \rrbrace_{X}}}\cdot
  \frac{\llbrace \phi \rrbrace_{X}(t_2)}{
    \totalweight{\llbrace \phi \rrbrace_{X}}}\cdot
  \min(\totalweight{\llbrace \psi \rrbrace_{X}},\totalweight{\llbrace
  \phi \rrbrace_{X}}) & \text{if } t = t_1 \pepasync{L} t_2\\
  0 & \text{otherwise}
\end{cases}
\]

Each process is interpreted as a weight function over process terms.
This is achieved by a Dirac-like function assigning $+\infty$ to the
$\Sigma$-term composed by the aforementioned variable: $\llbrace
P \rrbrace_{X}(t) = +\infty$ if $P=t$, 0 otherwise.  The
infinite rate characterizes instantaneous actions as if all the mass
is concentrated in the variable; e.g., in interactions based on the
minimal rate law, processes are not consumed.  For the same reason, if
we were dealing with concentration rates and the multiplicative law, we
would assign $1$ to $P$.

The remaining rules for hiding
are straightforward:
\[
\frac{P \xrightarrowu{a} \phi}
{P \setminus L \xrightarrowu{a} \phi}
\  a \notin L
\qquad
\frac{P \xrightarrowu{a} \phi}
{P \setminus L \xrightarrowu{\tau} \phi}
\ a \in L
\]

This completes the definition of $\llbrace\hole\rrbrace$
by structural recursion and hence the WF-GSOS specification of PEPA.  It
is easy to check that the induced ULTraS is functional and correspond
to the stochastic system of PEPA processes, that bisimulations on it
are stochastic bisimulations (and vice versa) and that bisimilarity is
a congruence with respect to the process signature.

\subsection{Segala-GSOS}
In \cite{bartels04thesis}, Bartels proposed a GSOS specification
format\footnote{Segala-GSOS specifications yield distributive
  laws for Segala systems but
  it still is an open problem whether every such distributive law
  arises from some Segala-GSOS specification.} for Segala systems
(hence Segala-GSOS), i.e.~ULTraS where weight functions are exactly
probability distributions. We recall Bartels' definition, with minor
notational differences.%
\begin{definition}[{\cite[§5.3]{bartels04thesis}}]
A \emph{GSOS rule for Segala systems} is a rule of the form
\vspace{-1ex}\[
\vspace{-1ex}
\frac{\left\{x_i \xrightarrow{a} \phi^a_{ij}\right\}_{
  1 \leq i \leq n,\ 
  a \in A_i,\ 
  1\leq j \leq m^a_i
}
\quad 
\left\{x_i \centernot{\xrightarrow{b}}\right\}_{
  1 \leq i \leq n,\ 
  b \in B_i
}
\quad 
\left\{\phi^a_{ij} \rightarrows y_k\right\}_{
  1 \leq k \leq q
}}
{\mathtt{f}(x_1,\dots,x_n) \xrightarrow{c} w_1\cdot t_1 + \dots + w_m \cdot t_m}
\]

\vspace{-.5ex}where:
\begin{itemize}\itemsep=-1pt
\item 
  $\mathtt f$ is an $n$-ary symbol from $\Sigma$;
\item 
  $X = \{x_i\mid 1 \leq i \leq n\}$, $Y = \{y_k\mid 1 \leq k \leq q\}$, and 
  $V = \{\phi^a_{ij}\mid 1 \leq i \leq n,\ a \in A_i,\ 1\leq j \leq m^a_i\}$
  are pairwise distinct \emph{process} and \emph{probability distribution} variables
  respectively;
\item 
  $a,b,c \in A$ are labels and $A_i \cap B_i = \emptyset$ for any 
  $i \in \{1,\dots ,n\}$;
\item
  $t_1,\dots, t_m$ are target terms on variables
  $X$, $Y$ and $V$; the latter are associated with colours from
  a finite palette to indicate different instances;
\item
  $\{w_i \in (0,1] \mid 1 \leq i \leq m\}$ describe a linear composition of the
  targets terms i.e.~are weights associated
  to the target terms and such that $w_1 + \dots + w_m = 1$.
\end{itemize}
A rule like above is \emph{triggered} by a tuple $\langle
C_1,\dots,C_n\rangle$ of \emph{enabled labels} if, and only if, $A_i \subseteq C_i$ and $B_i \cap C_i = \emptyset$ for each $i \in \{1,\dots ,n\}$.  A \emph{GSOS
  specification for Segala systems} is a set of rules in the above
format containing finitely many rules for any source symbol $\mathtt f$,
conclusion label $c$ and trigger $\vec C$.
\end{definition}

Segala-GSOS specifications can be easily turned into WF-GSOS ones.  The
first two families of premises are translated straightforwardly to the
corresponding ones in our format; the third can be turned into those
of the form $\support{\phi} \ni y$.  Targets of transitions describe
finite probability distributions and are evaluated to actual
probability distributions by a fixed interpretation of a form similar
to Definition~\ref{def:wf-gsos-eval}.  Some care is needed to
handle copies of probability variables.
In practice, duplicated variables are expressed by adding
``colouring'' operators to $\Theta$; their number is finite
and depends only on the set of rules since multiplicities
are fixed and finite for rules in the above format.
Let $\tilde V$ be the set of ``coloured'' variables from $V$ where the colouring is
used to distinguish duplicated variables (cf. \cite[§5.3]{bartels04thesis}).
Given a substitution $\nu$ from $\tilde V$ to (finite) probability distributions 
over $T^\Sigma(X + Y)$, each $t_i$ is interpreted as the probability distribution:
\vspace{-.5ex}\[
\vspace{-.5ex}
   \tilde t_i(t) \defeq 
   \begin{cases}
     \prod^{|\tilde V \cap var(t_i)|}_{k = 1} \nu(\phi_{k})(t_k) &
       \text{if $t = t_i[\phi_{k}/t_k]$ for $t_k \in T^\Sigma(X+Y)$} \\
     0 & \text{otherwise}
   \end{cases}
\]
and each target term $w_1\cdot t_1 + \dots + w_m \cdot t_m$ is
interpreted as the convex combination of $\tilde t_1,\dots,\tilde t_m$.

\subsection{Weighted GSOS}
In \cite{ks2013:w-s-gsos}, Klin and Sassone proposed a GSOS
format\footnote{Weighted GSOS specifications are proved to yield GSOS
  distributive laws for Weighted LTSs but it is currently an open
  question whether the format is also complete.  }  for Weighted LTSs
that is parametric in the commutative monoid $\mathfrak W$ and hence called
$\mathfrak W$-GSOS. The format subsumes many known formats for systems
expressible as $WLTS$: for instance, Stochastic GSOS specifications
are in the $\mathbb R^+_0$-GSOS format and GSOS for LTS are in the $\mathbb
B$-GSOS format where $\mathfrak 2 = (\{\ltrue,\lfalse\},\lor,\lfalse)$.
\begin{definition}[{\cite[Def.~13]{ks2013:w-s-gsos}}]\label{def:wgsos-rule}
A $\mathfrak W$-GSOS rule is an expression of the form: 
\vspace{-1ex}\[
\vspace{-1ex}
\frac{\left\{x_i \xrightarrowt{a} w_{ai}\right\}_{
  1 \leq i \leq n,\ 
  a \in A_i
}
\quad 
\left\{x_{i_k} \xrightarrowu{b_k,u_k} y_{k}\right\}_{
  1 \leq k \leq m
}}
{\mathtt{f}(x_1,\dots,x_n) \xrightarrowu{c,\beta(u_1,\dots,u_m)} t}
\]
where:
\vspace{-.5ex}
\begin{itemize}\itemsep=-1pt
\item 
  $\mathtt f$ is an $n$-ary symbol from $\Sigma$;
\item 
  $X = \{x_i\mid 1 \leq i \leq n\}$,
  $Y = \{y_k\mid 1 \leq k \leq m\}$ and
  $\{u_k\mid 1 \leq k \leq m\}$
  are pairwise distinct \emph{process} and 
  \emph{weight} variables;
\item
  $\{w_{ai} \in \mathfrak W \mid 1 \leq i \leq n,\ a\in A_i\}$ are \emph{weight constants}
  such that $w_{i_k} \neq 0$ for $1 \leq k \leq m$;
\item
  $\beta : W^m \to W$ is a multiadditive function on $\mathfrak W$;
\item 
  $a,b,c \in A$ are labels and $A_i \subseteq A$ for  
  $1 \leq i \leq n$;
\item
  $t$ is a $\Sigma$-term such that $Y \subseteq \mathrm{var}(\mathtt t) \subseteq X \cup Y$;
\end{itemize}
A rule is \emph{triggered} by a $n$-tuple $\vec C$ of \emph{enabled
labels} s.t.~$A_i \subseteq C_i$ and by a family of weights
$\{v_{ai}\mid 1\leq i \leq n,\ a \in A_i\}$ s.t.~$w_{ai} = v_{ai}$.  A
$\mathfrak W$-GSOS specification is a set of rules in the above format such
that there are only finitely many rules for the same source symbol,
conclusion label and trigger.
\end{definition}

Each rule describes the weight of $\mathtt t$
in terms of weights assigned to each $y_k$ (i.e.~$u_k$) 
occurring in it; if two rules share the same symbol, label, trigger and 
target then their contribute for $\mathtt t$ is added. 

To turn a $\mathfrak W$-GSOS specification into WF-GSOS ones, the first
step is to make weight function explicit, by means of premises like
\vspace{-2pt}$x_i \xrightarrowu{a} \phi^a_i$ (since WLTS are
functional ULTraS, i.e.~$m^a_i = 1$).  Then, each premise $x_i
\xrightarrowt{a} w_{ai}$ is translated into $\totalweight{\phi^a_i}
= w_{ai}$.  If $\mathfrak W$ is \emph{positive} (i.e., whenever $a+b=0$
then $a=b=0$) then $W \setminus\{0\}$ is a club and the
translation of a $\mathfrak W$-GSOS into a WF-GSOS is straightforward.
More generally, it suffices to combine rules sharing the same source,
label and trigger into a single WF-GSOS rule with the same source,
label and trigger. Its target is a suitable weight term containing the
functions $\beta$ and targets $\mathtt t$ of the original rules; every
occurrence of variables $y_k$ and $u_k$ is replaced with the
corresponding function variable (i.e.~$\phi^{b_k}_{i_k}$).  In
order to deal with multiple copies of the same weight variable, we
wrap each occurrence in a different ``colouring'' operator, like in
the case of Segala-GSOS.

\section{A coalgebraic presentation of ULTraS and WF-GSOS}\label{sec:coalg}
The aim of this section is to prove some important results about
WF-GSOS specifications.  We first provide a characterization of
ULTraSs as coalgebras for a specific behavioural functor
(Section~\ref{sec:ultras-as-coalgebras}), and their bisimulations as
\emph{cocongruences}.  Then, leveraging this characterization in
Section~\ref{sec:soundness} we apply Turi and Plotkin's bialgebraic
theory \cite{tp97:tmos}, which allows us to define the categorical
notion of \emph{WF-GSOS distributive laws}; these laws describe the
interplay between syntax and behaviour in any GSOS presentation of
ULTraS.  We will prove that every WF-GSOS specification yields a
WF-GSOS distributive law, i.e., the format is \emph{sound}.  As a
consequence, we obtain that the bisimilarities induced by these
specifications are always congruence relations.  Finally, in
Section~\ref{sec:completeness} we prove that WF-GSOS specification are
also \emph{complete}: every abstract WF-GSOS distributive law can be
described by means of a WF-GSOS specification.

\subsection{Abstract GSOS}\label{sec:abstract-gsos}
In \cite{tp97:tmos}, Turi and Plotkin detailed an abstract
presentation of well-behaved structural operational
semantics for systems of various kinds. There syntax and 
behaviour of transition systems are modelled by algebras 
and coalgebras respectively. For instance, an (image-finite)
LTS with labels in $A$ and states in $X$ is seen as
a (successor) function $h : X \to (\pf X)^A$ mapping
each state $x$ to a function yielding, for each label $a$,
the (finite) set of states reachable from $x$
via $a$-labelled transitions i.e.~$\{ y \mid x \xrightarrow{a} y\}$:
\vspace{-0.5ex}
\[
  y \in h(x)(a) \iff x \xrightarrow{a} y\text.
\vspace{-0.5ex}
\]
Functions like $h$ are \emph{coalgebras} for the 
(finite) \emph{labelled powerset functor} $(\pf)^A$
over the category of sets and functions $\cat{Set}$.
In general, state based transition systems can be viewed 
as \emph{$B$-coalgebra} i.e.~sets (\emph{carriers})
enriched by functions (\emph{structures}) like $h : X\to BX$
for some suitable covariant functor $B : \cat{Set} \to \cat{Set}$. 
The $\cat{Set}$-endofunctor $B$ is often called \emph{behavioural} 
since it encodes the computational behaviour characterizing the 
given kind of systems. A \emph{morphism} from a $B$-coalgebra $h : X \to BX$ to 
$g : Y \to BY$ is a function $f : X \to Y$ such that
the coalgebra structure $h$ on $X$ is consistently
mapped to the coalgebra structure $g$ on $Y$ 
i.e.~$g\circ f = Bf \circ h$.
Therefore, $B$-coalgebras and their homomorphisms form the category $\coalg{B}$.

Two states $x,y \in  X$ are said to be \emph{behaviourally equivalent}
with respect to the coalgebraic structure $h : X \to BX$ if
they are equated by some coalgebraic morphism from $h$. Behavioural
equivalences are generalised to two (or more) systems in the form
of kernel bisimulations \cite{staton11} 
i.e.~as the pullbacks of morphisms extending
to a cospan for the $B$-coalgebas structures associated with the
given systems as pictured below.
\vspace{-.5ex}\[\vspace{-.5ex}
\begin{tikzpicture}[auto,font=\small,yscale=1.2,xscale=1.5,
	baseline=(current bounding box.center)]
		\node (n0) at (0,1) {\(X_1\)};
		\node (n1) at (2,1) {\(X_2\)};
		\node (n2) at (1,.5) {\(Y\)};
		\node (n3) at (0,0) {\(B X_1\)};
		\node (n4) at (2,0) {\(B X_2\)};
		\node (n5) at (1,-.5) {\(B Y\)};
		\node (n6) at (1,1.5) {\(R\)};
		\draw[->] (n0) to node [swap] {\(f_1\)} (n2);
		\draw[->] (n1) to node [] {\(f_2\)} (n2);
		\draw[->] (n0) to node [swap] {\(h_1\)} (n3);
		\draw[->] (n1) to node [] {\(h_2\)} (n4);
		\draw[->] (n2) to node [swap] {\(g\)} (n5);
		\draw[->] (n3) to node [swap] {\(Bf_1\)} (n5);
		\draw[->] (n4) to node [] {\(Bf_2\)} (n5);
		\draw[->] (n6) to node [swap] {\(p_1\)} (n0);
		\draw[->] (n6) to node [] {\(p_2\)} (n1);
	\begin{scope}[shift=($(n6)!0.28!(n2)$),scale=0.4]
		\draw +(-.5,.25) -- +(0,0)  -- +(.5,.25);
	\end{scope}
\end{tikzpicture}
\]
If the cospan $f_1,f_2$ is jointly epic, 
i.e.~$j\circ f_1 = k \circ f_2 \implies j = k$ for any $j,k : C \to Z$,
(in general if $\{f_i\}$ is an epic sink, hence $\{p_i\}$ is a monic source)
then the set $Y$ is isomorphic to the equivalence classes induced by $R$.
We refer the interested reader to \cite{rutten:universal}
for more information on the coalgebraic approach to process theory.

Dually, process syntax is modelled via algebras for endofunctors.
Every algebraic signature $\Sigma$ defines an endofunctor
$\Sigma X = \coprod_{\mathtt f\in \Sigma}
X^{ar(\mathtt f)}$ on $\cat{Set}$ such that every model for 
the signature is an algebra for the functor i.e.~a set $X$ (carrier) together
with a function $g : \Sigma X \to X$ (structure). A morphism 
from a $\Sigma$-algebras $g :  \Sigma X \to X$ to
$h : \Sigma Y \to Y$ is a function $f : X \to Y$ such that $f \circ g = h \circ \Sigma f$.
The set of $\Sigma$-terms with variables from a set $X$ is denoted
by $T^\Sigma X$ and the set of ground ones admits an obvious
$\Sigma$-algebra $a : \Sigma T^\Sigma\emptyset \to T^\Sigma\emptyset$
which is the \emph{initial $\Sigma$-algebra} in the sense that
for every other $\Sigma$-algebra $g$, there exists a unique morphism
from $a$ to $g$ i.e.~the \emph{inductive extension} of the underlying
function $f : T^\Sigma\emptyset \to X$. The construction $T^\Sigma$
is a functor, moreover, it is the monad freely generated by $\Sigma$.

In \cite{tp97:tmos}, Turi and Plotkin showed that structural
operational specifications for LTSs in the well-known image
finite GSOS format \cite{bloomIM:95} correspond to
natural transformations of the following form: 
\[\lambda : \Sigma(\id \times B) \Longrightarrow BT^\Sigma\text.\]
\looseness=-1
These transformations, hence called \emph{GSOS distributive laws}, contain the information
needed to connect $\Sigma$-algebra and $B$-coalgebra
structures over the same carrier set
and capture the interplay between syntax and dynamics
at the core of the SOS approach. 
These structures are called \emph{$\lambda$-bialgebras} and are
formed by a carrier $X$ endowed with a $\Sigma$-algebra $g$ and a $B$-coalgebra
$h$ structure s.t.:
\[\vspace{-.5ex}
\begin{tikzpicture}[auto,font=\small,yscale=1.2,xscale=1.5,
	baseline=(current bounding box.center)]
		\node (n0) at (0,1) {\(\Sigma X\)};
		\node (n1) at (0,0) {\(\Sigma(X \times BX)\)};
		\node (n2) at (2,0) {\(BT^\Sigma X\)};
		\node (n3) at (1,1) {\(X\)};
		\node (n4) at (2,1) {\(BX\)};
		\draw[->] (n0) to node [] {\(g\)} (n3);
		\draw[->] (n0) to node [swap] {\(\Sigma \langle id_X, h\rangle\)} (n1);
		\draw[->] (n1) to node [] {\(\lambda_X\)} (n2);
		\draw[->] (n2) to node [swap] {\(Bg^\flat\)} (n4);
		\draw[->] (n3) to node [] {\(h\)} (n4);
\end{tikzpicture}
\vspace{-.5ex}\]
where $g^\flat: T^\Sigma X \to X$ is the canonical extension of $g$
by structural recursion.
In particular, every $\lambda$-distributive law gives rise to a $B$-coalgebra structure
over the set of ground $\Sigma$-terms $T^\Sigma\emptyset$ and to a
$\Sigma$-algebra structure on the carrier of the final $B$-coalgebra.
These two structures are part of the initial and final 
$\lambda$-bialgebra respectively and therefore, because the unique
morphism from the former to the latter is both a $\Sigma$-algebra
and a $B$-coalgebra morphism, observational
equivalence on the system induced over $T^\Sigma\emptyset$
is a congruence with respect to the syntax $\Sigma$.

\subsection{ULTraSs as coalgebras}\label{sec:ultras-as-coalgebras}
Since ULTraSs alternate non-deterministic steps with quantitative
steps, the corresponding behavioural functor can be obtained by
composing the usual functor $(\pf)^A: \set \to \set$ of non-deterministic labelled
transition systems with the functors capturing the quantitative
computational aspects: $\fw$. 
It is easy to see that the action of a set function on a weight function
\eqref{eq:fw-action} preserves identities and composition rendering
$\fw$ an endofunctor over \set.

For any $\mathfrak W$ and any $A$, $A$-labelled image-finite $\mathfrak W$-ULTraSs
and their homomorphisms clearly form a category: $\cat{ULTS}_{\mathfrak W,A}$.
Objects and morphisms of this category are in 1-1 correspondence
with $(\pf\fw)^A$-coalgebras and their homomorphisms respectively.
\begin{proposition}\label{prop:ultras-as-coalgebras}
	$\cat{ULTS}_{\mathfrak W,A} \cong \coalg{(\pf\fw)^A}$.
\end{proposition}
\begin{proof}
	Any image-finite $\mathfrak W$-ULTraS $(X,A,\rightarrowu)$ determines a
	coalgebra $(X,h)$ where, for any $x \in X$ and $a \in A$: $h(x)(a)
	\defeq \{\rho \mid x\xrightarrowu{a} \rho\}$.  Image-finiteness
	guarantees that these sets are finite and that their elements are
	finitely supported weight functions from $X$ to the carrier of $\mathfrak
	W$.  Then, it is easy to check that the correspondence is bijective.
\end{proof}

A similar result holds for the bisimulation given in
Definition~\ref{def:bisim}.  Categorically, a relation between $X$ and
$Y$ is a (jointly monic) span $X\leftarrow R \rightarrow Y$.  In our
case, this span has to be subject to some conditions, as shown 
next.
\begin{proposition}\label{prop:bisim-correspondence}
  Let $(X_1,A,\rightarrowu_1)$ and $(X_2,A,\rightarrowu_2)$ be
  two image-finite
  $\mathfrak W$-ULTraSs; let $(X_1,h_1)$, $(X_2,h_2)$ be the corresponding
  coalgebras according Proposition~\ref{prop:ultras-as-coalgebras}.
  A relation between $X_1$ and $X_2$ is a bisimulation iff there exists a coalgebra $(Y,g)$ and
  two coalgebra morphisms $f_1:(X_1,h_1)\to (Y,g)$ and
  $f_2:(X_2,h_2)\to (Y,g)$ such that $f_1,f_2$ are jointly epic and
  $R$ is their pullback, i.e.~the
  diagram below commutes.
\[\begin{tikzpicture}[auto,font=\footnotesize,yscale=1.5,xscale=1.96,
	baseline=(current bounding box.center)]
		\node (n0) at (0,1) {\(X_1\)};
		\node (n1) at (2,1) {\(X_2\)};
		\node (n2) at (1,.5) {\(Y\)};
		\node (n3) at (0,0) {\((\pf\fw X_1)^A\)};
		\node (n4) at (2,0) {\((\pf\fw X_2)^A\)};
		\node (n5) at (1,-.5) {\((\pf\fw Y)^A\)};
		\node (n6) at (1,1.5) {\(R\)};
		\draw[->] (n0) to node [swap] {\(f_1\)} (n2);
		\draw[->] (n1) to node [] {\(f_2\)} (n2);
		\draw[->] (n0) to node [swap] {\(h_1\)} (n3);
		\draw[->] (n1) to node [] {\(h_2\)} (n4);
		\draw[->] (n2) to node [swap] {\(g\)} (n5);
		\draw[->] (n3) to node [swap] {\((\pf\fw f_1)^A\)} (n5);
		\draw[->] (n4) to node [] {\((\pf\fw f_2)^A\)} (n5);
		\draw[->] (n6) to node [swap] {\(p_1\)} (n0);
		\draw[->] (n6) to node [] {\(p_2\)} (n1);
	\begin{scope}[shift=($(n6)!0.28!(n2)$),scale=0.4]
		\draw +(-.5,.25) -- +(0,0)  -- +(.5,.25);
	\end{scope}
\end{tikzpicture}\]
\end{proposition}
\begin{proofatend}
Let $\mathfrak W = (W,+,0)$ be a commutative monoid and let $(X,A,\rightarrowu_X)$,
$(Y,A,\rightarrowu_Y)$, $(X,\alpha)$ and $(Y,\beta)$ be two ULTraS
over $\mathfrak W$ and their corresponding coalgebras 
(Proposition~\ref{prop:ultras-as-coalgebras}).
Recall that a function $f : X \to Y$ is a also 
coalgebra morphism $f : \alpha \to \beta$ iff, 
for each $x \in X$, and $a \in A$:
\[ f(x) \xrightarrowu{a}_Y \psi \iff x \xrightarrowu{a}_X \phi 
\land \psi = \phi[f] \]
where $\phi[f]$ denotes the action of $f$ on $\phi$ (i.e.~the
function $\lambda y : Y . \sum_{x\in f^{-1}(y)}\phi(x)$)
and function equality is defined point-wise as usual.
Firstly, we prove that if $R$ is a kernel relation
of some jointly epic cospan of coalgebra mophism from $\alpha$ and $\beta$
then it is a bisimulation.
Let the aforementioned cospan be 
$(X,\alpha) \xleftarrow{f} (Z,\gamma) \xrightarrow{g} (Y,\beta)$,
$(Z,A,\rightarrowu_Z)$ the ULTraS for $\gamma$
and assume $x$ and $y$ such that $f(x) = g(y)$.
By definition of coalgebra morphism, $f(x) = z$ implies:
\[x \xrightarrowu{a}_X \phi \iff
  z \xrightarrowu{a}_Z \rho = \phi[f] = \lambda c:Z\sum_{x\in f^{-1}(c)}\phi(x)
\text.\]
Likewise $g(y) = z$ implies:
\[
  y \xrightarrowu{a}_Y \psi \iff
  z \xrightarrowu{a}_Z \rho = \psi[g] = \lambda c:Z\sum_{y\in g^{-1}(c)}\psi(y)
\text.\]
Therefore $f(x) = g(y)$ implies:
\begin{gather*}
  x \xrightarrowu{a}_X \phi \implies y \xrightarrowu{a}_Y \psi \land
\forall C \in Z.Z\sum_{x\in f^{-1}(C)}\phi(x) =  Z\sum_{y\in g^{-1}(C)}\psi(y)\\
y \xrightarrowu{a}_Y \psi \implies x \xrightarrowu{a}_X \phi \land
\forall C \in Z.Z\sum_{x\in f^{-1}(C)}\phi(x) =  Z\sum_{y\in g^{-1}(C)}\psi(y)
\end{gather*}
Then, we conclude by noting that if $R$ is the kernel of $f,g$
there is a bijective correspondence between its equivalence classes
and elements in $Z$ since every class is in the image of $f$ or $g$
by the jointly epic assumption.

For the converse, given a bisimulation $R$ for $(X,A,\rightarrowu_X)$
$(Y,A,\rightarrowu_Y)$ let $Z$ be the set of the equivalence classes in $R$ and
consider the ULTraS $(Z,A,\rightarrowu_Z)$ defined as follows:
\begin{gather*}
  C\xrightarrowu{a}_Z \lambda D:Z.\sum_{x' \in D} \phi(x') \iff x\xrightarrow{a}_X \phi \land x \in C \\
  C\xrightarrowu{a}_Z \lambda D:Z.\sum_{y' \in D} \psi(y') \iff y\xrightarrow{a}_Y \psi\land y \in C 
\end{gather*}
The two statements are redundant since $x,y \in C \iff x R y$ and hence
iff  for every $x\xrightarrow{a}_X\phi$ there is 
$y \xrightarrow{a}_Y \psi$ s.t.~$\phi \equiv_R\psi$ and vice versa.
Finally, class membership defines a jointly epic coalgebra cospan
from the coalgebras associated to $(X,A,\rightarrowu_X)$ and
$(Y,A,\rightarrowu_Y)$ to the one associated to $(Z,A,\rightarrowu_Z)$
by simply mapping each $x \in X$ and each $y \in Y$ to its class.
\end{proofatend}
Intuitively, the system $(Y,g)$ ``subsumes'' both $(X_1,h_1)$ and
$(X_2,h_2)$ via $f_1$, $f_2$; then, $R$ relates the states which are
mapped to the same behaviour in $Y$ ($(x_1,x_2)\in R$ iff
$f_1(x_1)=f_2(x_2)$).  

\paragraph{Coalgebraic bisimulation} In Concurrency Theory also
Aczel-Medler's \emph{coalgebraic bisimulation} \cite{am89:final} is
widely used.
In fact, it is known that kernel bisimulations and coalgebraic
bisimulations coincide if the behavioural functor is \emph{weak
  pullback preserving} (wpp).  This is the case for many behavioural
functors, but not for $\fw$ in general
\cite{ks2013:w-s-gsos}. Actually, the fact that $\fw$ (and
$(\pf\fw)^A$) preserves weak pullbacks depends on the underlying monoid only.
\begin{definition}
	A commutative monoid is called \emph{positive} (sometimes zerosumfree, positively ordered) whenever $x + y = 0 \implies x = y = 0$ holds true.
	It is called  \emph{refinement} if for each $r_1 + r_2 = c_1 + c_2$
	there is a $2\times 2$ matrix $(m_{i,j})$ s.t.~
	$r_i = m_{i,1} + m_{i,2}$ and $c_j = m_{1,j} + m_{2,j}$.
\end{definition}
\begin{lemma}
	\label{lem:wpp}
	Coalgebraic bisimulation and behavioural equivalence on ULTraSs
	coincides if $\mathfrak W$ is a positive refinement monoid.
\end{lemma}
\begin{proof} 
  $(\pf)^A$ is wpp, and under the lemma hypothesis also $\fw$ is wpp,
  by \cite{gummS:01monlbl}.  Therefore, $(\pf\fw)^A$ is wpp, hence
  every behavioural equivalence is a coalgebraic bisimulation on
  $(\pf\fw)^A$-coalgebras.  We conclude by
  Proposition~\ref{prop:ultras-as-coalgebras}.
\end{proof} 
This condition, can be easily verified and in fact holds for several
monoids of interest, e.g.: $(\{\ltrue,\lfalse\},\lor,\ltrue)$, $(\mathbb N,+,0)$,
$(\mathbb R^+_0,+,0)$, $(\mathbb N,\max,0)$, and $(A^*,\cdot,\varepsilon)$.
A simple counter example is $(\{0,a,b,1\}, +, 0)$ where
$x + y \defeq 1$ whenever $x \neq 0 \neq y$
for it is positive but not refinement 
(cf.~$\vec r = \langle a, a\rangle$ and $\vec c = \langle b, b\rangle$).

\vspace{-1ex}
\subsection{WF-GSOS specifications are WF-GSOS distributive laws}
\label{sec:soundness}\label{sec:cong-proof}
In this subsection we put the WF-GSOS format within the bialgebraic
framework \cite{tp97:tmos}.  As a consequence, we obtain that the
bisimilarity induced by the ULTraS defined by this specification is a
congruence.

In particular, we prove that every WF-GSOS specification represents a
distributive law of the signature over the ULTRaS
behavioural functor, i.e., a natural transformation of the form
\vspace{-.5ex}\begin{equation}\label{eq:wf-gsos-nat}\vspace{-.5ex}
  \lambda : 
  \Sigma(\id \times (\pf\fw)^A) 
  \Longrightarrow
  (\pf\fw T^\Sigma)^A
\end{equation}
where $A$ is the set of labels, $\mathfrak W$ is the commutative monoid
of weights,
$\Sigma = \coprod_{\mathtt f\in \Sigma}\id^{\mathrm{ar}(\mathtt f)}$ is
the syntactic endofunctor induced by the process signature $\Sigma$,
and $T^\Sigma$ is the free monad for $\Sigma$.  We will call natural
transformations of this type \emph{WF-GSOS distributive laws}.

Before stating the soundness theorem, we note that every natural
transformation $\lambda$ as above induces a
$(\pf\fw)^A$-coalgebra structure over ground
$\Sigma$-terms. Namely, this is the only function $h_\lambda :
T^\Sigma\emptyset \to (\pf\fw(T^\Sigma\emptyset))^A$ such that:\vspace{-.5ex}
\begin{equation}\label{eq:induced-coalg}\vspace{-.5ex}
h_\lambda \circ a = 
(\pf\fw T^\Sigma(a^\#))^A\circ\lambda_X \circ \Sigma\langle id,h_\lambda\rangle
\end{equation}
where $a^\# : T^\Sigma T^\Sigma\emptyset \to T^\Sigma\emptyset$ is the
inductive extension of $a$. 

We can now provide the soundness result for WF-GSOS specifications
with respect to WF-GSOS distributive laws, and between
systems and coalgebras they induce over ground $\Sigma$-terms.

\begin{theorem}[Soundness]\label{thm:wf-gsos-soundness}
A specification $\langle\mathcal R,\llbrace\hole\rrbrace\rangle$
yields a natural transformation $\lambda$ as in \eqref{eq:wf-gsos-nat}
such that $h_\lambda$ and the ULTraS induced by $\langle\mathcal R,
\llbrace\hole\rrbrace\rangle$ coincide.\par
\end{theorem}
\begin{proof}
	For any set $X$, define the function $\lambda_X$ as the composite:
	\vspace{-.5ex}\[\vspace{-.5ex}
	  \Sigma(X \times (\pf\fw X)^A) 
	  \xrightarrow{\llbracket\mathcal R\rrbracket_X}
	  (\pf T^\Theta(X + \fw X))^A 
	  \xrightarrow{(\mu\circ\pf\llbrace\hole\rrbrace_X)^A\!\!}
	  (\pf\fw T^\Sigma X)^A
	\]
	where $\mu : \pf\pf \Rightarrow \pf$ and $\llbracket\mathcal R\rrbracket_X$ is defined as follows: for all
	$\psi' \in T^\Theta(X + \fw X)$, $\mathtt f \in \Sigma$, $c \in A$,
	trigger $\vec A = \langle A_1,\dots A_n\rangle$,
	$\vec w = \langle w_1,\dots w_p\rangle$,
	$y'_k\in X$ and 
	$\Phi_i(a) = \{\phi^a_{ij} \in \fw X \mid 1\leq j \leq m^a_i\}$ for 
	$n = \mathrm{ar}(\mathtt f)$ and $i \in \{1,\dots,n\}$,
	let
	\vspace{-.5ex}\[\vspace{-.5ex}
	  \psi' \in \llbracket\mathcal R\rrbracket_X
	  (\mathtt f((x'_1,\Phi_1),\dots,(x'_n,\Phi_n)))\]
	if, and only if, there exists in $\mathcal R$ a (possibly renamed) rule
	\vspace{-1ex}\[\vspace{-1ex}\frac{
		\begin{array}{c}
			\Big\{
			x_i \xrightarrowu{a} \phi^a_{ij}
			\Big\}
			\hspace{-1.2ex}\begin{array}{l}
				\scriptstyle  1 \leq i \leq n,\\[-4pt]
				\scriptstyle  a \in A_i,\\[-4pt]
				\scriptstyle  1\leq j \leq m^a_i
			\end{array}
			\quad
			\Big\{
			x_i \centernot{\xrightarrowu{b}}
			\Big\}
			\hspace{-1.2ex}\begin{array}{l}
				\scriptstyle  \\[-4pt]
				\scriptstyle  1 \leq i \leq n,\\[-4pt]
				\scriptstyle  b \in B_i
			\end{array}
			\quad
			\Big\{
			\totalweight{\phi^{a_k}_{i_kj_k}} = w_k
			\Big\}
			_{1 \leq k \leq p}
			\quad
			\Big\{
			\corestr{\phi^{a_k}_{i_kj_k}}{\mathfrak C_k}\ni y_k
			\Big\}
			_{1 \leq k \leq q}
		\end{array}
	}{
		\mathtt{f}(x_1,\dots,x_n) \xrightarrowu{c} \psi
	}\]
	such that $m^a_i \neq 0$ iff $a \in A_i$ and there exists
	a substitution $\sigma$ such that $\psi' = \sigma[\psi]$,
	$\sigma x_i = x'_i$, $\sigma y_k = y'_k$, $\sigma\phi^a_{ij} = \phi^a_{ij}$,
	$\totalweight{\phi^{a_k}_{i_kj_k}} = w_k$ and 
	$\phi^{a_k}_{i_kj_k}(\sigma y_k) \in \mathfrak C_k$.
	Then, naturality can be proved separately for the two components:
	the former can be tackled as in \cite[Th.~1.1]{tp97:tmos}
	and the latter readily follows from Definition~\ref{def:wf-gsos-eval}.
	
	Correspondence of $h_\lambda$ with the induced ULTraS follows by noting that
	the latter is given by structural recursion on $\Sigma$-terms
	by applying precisely $\lambda$ as given above
	(cf.~\eqref{eq:induced-coalg} and
	Definition~\ref{def:induced-ultras}).
\end{proof} 

Now, by general results from the bialgebraic framework, every
behavioural equivalence on $h_\lambda$ is also a congruence on
$T^\Sigma\emptyset$.  In order to obtain this result we need the
following (simple yet important) property.
\begin{proposition}\label{prop:final-coalg}
	The category of $(\pf\fw)^A$-coalgebras has a final object.
\end{proposition}
\begin{proof} 
  By \cite{barr93} every finitary $\set$ endofunctor admits a
  final coalgebra. By definition $\fw$ is finitary. The thesis follows from $\pf \cong \ff{2}$ and
  from finitarity being preserved by functor composition.
\end{proof} 

\begin{corollary}[Congruence]
  Behavioural equivalence on the coalgebra over $T^\Sigma\emptyset$
  induced by $\langle\mathcal R,\llbrace{\hole}\rrbrace\rangle$ is a
  congruence with respect to the signature $\Sigma$.
\end{corollary}
\begin{proof} 
The syntactic endofunctor $\Sigma$ admits an initial algebra and, by 
Proposition~\ref{prop:final-coalg}, the behavioural endofunctor
$(\pf\fw)^A$ admits a final coalgebra. The same holds for
their free monad and cofree copointed functor respectively.
The specification $\langle\mathcal R,\llbrace\hole\rrbrace\rangle$
defines, by Theorem~\ref{thm:wf-gsos-soundness}, a distributive law
which uniquely extends to a distributive law distributing the
free monad over the cofree copointed functor; then the thesis
follows from \cite[Cor.~7.3]{tp97:tmos}.
\end{proof} 

\vspace{-1ex}
\subsection{WF-GSOS distributive laws are WF-GSOS specifications}\label{sec:completeness}
In this subsection we give the important result that the WF-GSOS
format is also \emph{complete} with respect to distributive
laws of the form \eqref{eq:wf-gsos-nat}.

\begin{theorem}[Completeness]\label{thm:wf-gsos-completeness}
  Every WF-GSOS distributive law $\lambda$ arises from some WF-GSOS
  specification $\langle\mathcal R,\llbrace\hole\rrbrace\rangle$.
\end{theorem}

The proof of this Theorem follows the methodology introduced by
Bartels for proving adequacy of Bloom's GSOS specification format
\cite[§3.3.1]{bartels04thesis}.  The (rather technical) proof will
take the rest of this subsection, so for sake of conciseness we omit
to recall some results which can be found in \emph{loc.~cit.}.

The thesis follows from proving that, for every $\lambda$,
there exists an image-finite set of WF-SOS rules $\mathcal R$
(and suitable interpretations $\theta$ and $\xi$) 
making the diagram in Figure~\ref{fig:diag-comp} commute.
\begin{figure}
	\centering
	\begin{tikzpicture}[auto,xscale=2.4,yscale=1.2, font=\footnotesize,
		baseline=(current bounding box.center)]
		\node (n0) at (0,3) {$\Sigma(\id \times (\pf\fw)^A)$};
		\node (n1) at (4,3) {$(\pf\fw T^\Sigma)^A$};
		\node (n2) at (0,0) {$(\pf T^\Xi(\id + \fw))^A$};
		\node (n3) at (4,0) {$(\pf\fw T^\Sigma)^A$};
		\node (n4) at (1.2,2) {$(\pf T^\Theta(\id + \fw))^A$};
		\node (n5) at (2.8,2) {$(\pf^2\fw T^\Sigma)^A$};
		\node (n6) at (1.2,1) {$(\pf^2 T^\Theta(\id + \fw))^A$};
		\node (n7) at (2.8,1) {$(\pf^3\fw T^\Sigma)^A$};
		
		\draw[->] (n0) to node {$\lambda$} (n1);
		\draw[->] (n0) to node[swap] {$\llbracket\mathcal R\rrbracket$} (n2);
		\draw[->] (n0) to node {$\rho$} (n4);
		\draw[->] (n2) to node[swap] {$(\pf\llbrace\hole\rrbrace)^A$} (n3);
		\draw[->] (n2) to node[swap] {$(\pf\xi)^A$} (n6);
		\draw[->] (n3) to node[] {$(\mu_{\scriptscriptstyle \fw T^\Sigma})^A$} (n1);
		\draw[->] (n4) to node {$(\pf\theta)^A$} (n5);
		\draw[->] (n5) to node[pos=.2] {$(\mu_{\scriptscriptstyle \fw T^\Sigma})^A$} (n1);
		\draw[->] (n6) to node[] {$(\mu_{\scriptscriptstyle T^\Theta(\id + \fw)})^A$} (n4);
		\draw[->] (n6) to node[swap] {$(\pf^2\theta)^A$} (n7);
		\draw[->] (n7) to node[swap] {$(\pf\mu_{\scriptscriptstyle \fw T^\Sigma})^A$} (n3);
		\draw[->] (n7) to node {$(\pf\mu_{\scriptscriptstyle \fw T^\Sigma})^A$} (n5);
		
		\node[red, font=\small] at (2,2.6) {(Lem.~\ref{lem:wf-gsos-factorization})};
		\node[red, font=\small] at (2,.3) {(Def.~$\llbrace\hole\rrbrace$)};
		\node[red, font=\small] at (3.1,1.5) {($=$)};
		\node[red, font=\small] at (1.7,1.5) {(Nat.)};
	\end{tikzpicture}
	\caption{Factorization for $\lambda$-distributive laws as WF-GSOS specifications.}
	\label{fig:diag-comp}
\end{figure}
The lower part of the diagram defines the interpretation
$\llbrace\hole\rrbrace$ out of $\xi$ and $\theta$ completing
the WF-GSOS specification for $\lambda$.
The middle and right parts of the diagram trivially commute.

The upper part of the diagram commutes because of the following
lemma which states that every WF-GSOS distributive law
arises from an interpretation and a natural transformation
having the same type of those defined by image-finite
sets of WF-GSOS rules.
\begin{lemma}\label{lem:wf-gsos-factorization}
	Let $\Sigma$, $A$ and $\mathfrak W$ be a signature, a set of labels and
	a commutative monoid, respectively.
	Let $\lambda$ be a WF-GSOS distributive law as in \eqref{eq:wf-gsos-nat}.
	There exist $\Theta$ and an interpretation 
	factorizing $\lambda$ i.e.~there exists 
	$\rho : \Sigma(\id \times (\pf\fw)^A)\Rightarrow (\pf T^\Theta(\id + \fw))^A$ 
	such that $\lambda = (\mu \circ \pf \theta )^A \circ \rho$.
\end{lemma}
\begin{proof}[Proof (sketch)]
	In $\set$ it is easy to encode finitely supported functions
	as terms. For instance let $\Theta$ extend $\Sigma$ with 
	operators for describing collections and weight assignments (e.g.~$(\hole\mapsto w)$ where $w \in \mathfrak W \setminus \{0\}$).
	Then, we can turn $\lambda$ into $\rho$ by simply encoding its
	codomain. Then $\theta$ simply evaluates these terms
	back to weight functions everything else to the $\emptyset$.
\end{proof}
Following Bartels' methodology, the left part of the diagram commutes
by reducing $\rho$
to simpler, but equivalent, families of natural transformations and
eventually deriving a syntactical specification which is then shown to
be equivalent to an image-finite set of WF-GSOS rules and an
intermediate interpretation $\xi$.
The use of another signature $\Xi$
besides $\Theta$
gives us an extra degree of freedom and simplifies the proof.  In
particular, it allows us to encode natural transformations of type
$\fw
\Rightarrow \pf\fw$ (yielded by the aforementioned reduction) in
$\xi$
and handle them downstream to the interpretation
$\llbrace\hole\rrbrace$.
This expressiveness gain is one of the reasons for the introduction of
non-determinism in Definition~\ref{def:wf-gsos-eval}.

First, note that, by \cite[Lem.~A.1.1]{bartels04thesis}, $\rho$ as above is equivalent to:
\vspace{-.5ex}\[\vspace{-.5ex}
\textstyle
	\bar\rho:\Sigma(\id \times (\pf\fw)^A)\times A
	\Longrightarrow 
	\pf T^\Theta(\id +\fw)
\]
which is equivalent to a family of natural transformations
\vspace{-.5ex}\begin{equation}\label{eq:alpha-nat}\vspace{-.5ex}
\textstyle
	\alpha_{\mathtt f,c} : (\id \times (\pf\fw)^A)^N
	\Longrightarrow 
	\pf T^\Theta(\id + \fw)
\end{equation}
indexed by $\mathtt f \in \Sigma$ and $c \in A$ and 
where $N = \{1,\dots,\mathrm{ar}(\mathtt f)\}$.
In fact, $\Sigma$ is a polynomial functor and $\id \times A \cong A \cdot \id$
is an $|A|$-fold coproduct.

By \cite[Lem.~A.1.7]{bartels04thesis}, each $\alpha_{\mathtt f,c}$ is equivalent to a natural transformation
\vspace{-.5ex}\begin{equation}\label{eq:alpha-bar-nat}\vspace{-.5ex}
\textstyle
	  \bar\alpha_{\mathtt f,c} : 
	  (\pf\fw)^{A\times N}
	  \Longrightarrow \pf T^\Theta(N + \id + \fw)
\end{equation}
and, by the natural isomorphism
\vspace{-.5ex}\[\vspace{-.5ex}
\textstyle
	(\pf)^{A\times N} \cong (\pf^+ + 1)^{A\times N} \cong
	\coprod_{E \subseteq A \times N} (\pf^+)^E
\]
each $\bar \alpha_{\mathtt f,c}$ is equivalent to a family of
natural transformations
\vspace{-.5ex}\begin{equation}\label{eq:beta-nat}\vspace{-.5ex}
\textstyle
	\beta_{\mathtt f,c,E} : (\pf^+\fw)^{E}
	\Longrightarrow 
	\pf T^\Theta(N + \id + \fw)
\end{equation}
where the added index corresponds to the 
vector of sets of labels $\langle E_1,\dots,E_{\mathrm{ar}(\mathtt f)}\rangle$
composing the trigger of a WF-GSOS rule.
By the natural isomorphism
\vspace{-.5ex}\[\vspace{-.5ex}
\textstyle
	\pf^+\fw \cong \pf^+\coprod_{v \in \mathfrak W}\fw^v \cong
	\coprod_{V \in \pf^+\mathfrak W}
	\prod_{v \in V}\pf^+\fw^v
\]
where $\fw^v X \defeq \{ \phi \in \fw X \mid \totalweight{\phi} = v\}$,
each $\beta_{\mathtt f,c,E}$ is equivalent to a family of
natural transformations
\vspace{-.5ex}\begin{equation}\label{eq:gamma-nat}\vspace{-.5ex}
\textstyle
	  \gamma_{\mathtt f,c,E,w} : 
	  \coprod_{e \in E}\prod_{v \in w(e)}\pf^+\fw^v
	  \Longrightarrow 
	  \pf T^\Theta(N + \id + \fw)
\end{equation}
where $w : E \to \pf^+\mathfrak W$. 
Since total weight premises associate pairs from $E$
to weights, maps like $w$ can be seen as families 
of triggering weights.

By \cite[Lem.~A.1.3]{bartels04thesis} and by the natural isomorphism
\vspace{-.5ex}\[\vspace{-.5ex}
\textstyle
T^\Theta \cong \coprod_{\psi \in T^\Theta 1} \id^{|\psi|_*}
\]
where $|\psi|_*$ denotes the number of occurrences of $* \in 1$
in the $\Theta$-term $\psi$ (cf.~\cite[Lem.~A.1.5]{bartels04thesis})
each $\gamma_{\mathtt f,c,E,w}$ corresponds to a family
of natural transformations
\vspace{-.5ex}\begin{equation}\label{eq:delta-nat}\vspace{-.5ex}
\textstyle
	  \delta_{\mathtt f,c,E,w,\psi} : 
	  \coprod_{e \in E}\prod_{v \in w(e)}\pf^+\fw^v
	  \Longrightarrow 
	  \pf^+((\id + \fw)^{|\psi|_*})
\end{equation}
where the added index $\psi$ ranges over some subset of
$T^\Theta(1+N)$ (cf.~target terms of WF-GSOS rules).

Then, following \cite[§3.3.1, Cor.~A.2.8]{bartels04thesis} it is 
easy to check that each $\delta_{\mathtt f,c,E,w,\psi}$ describes a 
non-empty, finite set of derivation rules as
\vspace{-.5ex}\[\vspace{-.5ex}
	\frac{
		\phi_{j,{v_j}} \in \pi_{v_j}(\Phi_{e_j}) \quad y_i \in \epsilon_{j,{v_j}}(\phi_{j,{v_j}})
	}{
		\langle z_{1},\dots,z_{|\psi|_*}\rangle \in 
		\delta_{\mathtt f,c,E,w,\psi}((\Phi_e)_{e \in E})
	}
\]
where $p,q\in \mathbb N$, $e_j \in E$, $1\leq j \leq p$,
$1\leq i\leq q$, $v_j \in w(e_j)$, each $z_k \in \{y_i \mid 1\leq i\leq q \}$
for $1 \leq k \leq |\psi|_*$ and each $\epsilon_{j,v_j}$ is a natural transformation:
\vspace{-.5ex}\[\vspace{-.5ex}\epsilon_{j,{v_j}}:\fw^{v_j} \Longrightarrow \pf^+(\id + \fw)\text{.}\]
Natural transformations of this type can be easily encoded in the
term $\psi$ by suitable extensions of $\Theta$ and therefore
each $\delta_{\mathtt f,c,E,w,\psi}$ can be shown to be equivalent
to a $\delta$-specification i.e.~a non-empty, finite set of derivation rules as above
except for each $z_k$ being a term wrapping $\phi_{j,v}$ with
the symbol denoting $\epsilon_{j,v}$. These terms are then evaluated 
by the interpretation $\xi^\delta$ as expected.

This proof points out the trade-off that has to be made in presence of
specifications with interpretation such as WF-GSOS or MGSOS
\cite{bm:2015stocsos}.  In fact, clubs were not mentioned in the above
reduction of $\rho$ since each $\epsilon_{j,v_j}$ was handled by the
interpretation $\xi$. However, the following result shows that clubs
(hence, premises like $\corestr{\phi}{\mathfrak C}\ni y$), characterize
natural transformations of type $\fw^v \Rightarrow \pf$.
\begin{lemma}\label{lem:clubs-nat}
	For any natural transformation $\upsilon : \fw^w \Rightarrow \pf$
	there exists a club $\mathfrak C_\upsilon$ characterizing it:
	$x \in \upsilon_X(\phi) \iff \phi(x) \in \mathfrak C_\upsilon$.
\end{lemma}
\begin{proof}[Proof (sketch)]
	Intuitively, natural transformations of this type are ``selecting
	a finite subset from each weight function domain'' and it is
	easy to check that elements can be only singled out by their
	weight. Likewise, finiteness and naturality prevent the selection of
	anything outside function supports. Then, the problem readily
	translates into finding the finest topology on the weight monoid
	that ``plays well'' with $\fw$ i.e.~such that monoidal addition,
	seen as a continuous map from the product topology,
	preserves opens (i.e.~any admissible selection). Clubs are a base for this
	topology since, by definition, these are the only substructures
	isolated w.r.t.~$\fw$-action. Hence selections made by $\upsilon$ 
	are completely characterized by a single club $\mathfrak C_\upsilon$.
\end{proof}

Finally, we have to translate the set of rules we got so far into the
WF-GSOS format; we do it by reversing the chain that led us from
$\rho$ to $\delta$ and $\delta$-specification.  By
Lemma~\ref{lem:clubs-nat} every $\delta$-specification
is equivalent to a $\gamma$-specification
\vspace{-.5ex}\[\vspace{-.5ex}
	\left\{
	\frac{
		\phi_{j} \in \Phi_{e_j} \quad 
		\totalweight{\phi_j} = v_j \quad  
		\corestr{\phi_j}{\mathfrak C_i} \ni y_i
	}{
		\psi(z_1,\dots,z_{|\psi|_*}) \in 
		\gamma_{\mathtt f,c,E,w}((\Phi_e)_{e \in E})
	}\,\middle|
		\begin{array}{l}\scriptstyle
			v_j \in w(e_j),\\\scriptstyle
			z_k \in \{\phi_j[\zeta_j], y_i\}+N,\\\scriptstyle
			\psi
		\end{array}\!\!\!
	\right\}_\text{fininte}
\]
where $\phi_j[\zeta_j]$ is a term build with the $\Xi$-operator
denoting the natural transformation 
$\zeta_j : \fw^{v_j} \Rightarrow \pf\fw$ and $\xi^\gamma$ acts as
$\xi^\delta$ on these terms, as the identity on those generated from $\Theta$
(distributing the powerset as expected) and maps everything else to $\emptyset$.
A $\gamma$-specification defines a natural transformation as in \eqref{eq:gamma-nat}
and every family of $\gamma$-specifications characterizing
a natural transformation as in \eqref{eq:beta-nat} is equivalent to
a $\beta$-specification i.e.~a set of derivation rules
\vspace{-.5ex}\[\vspace{-.5ex}
	\left\{
	\frac{
		\phi_{j} \in \Phi_{e_j} \quad 
		\totalweight{\phi_j} = v_j \quad  
		\corestr{\phi_j}{\mathfrak C_i} \ni y_i
	}{
		\psi(z_1,\dots,z_{|\psi|_*}) \in 
		\beta_{\mathtt f,c,E}((\Phi_e)_{e \in E})
	}\,\middle|
		\begin{array}{l}\scriptstyle
			z_k \in \{\phi_j[\zeta_j], y_i\}+N,\\\scriptstyle
			\psi
		\end{array}\!\!\!
	\right\}_\text{image finite}
\]
finite up to vectors of
total weights $\vec v = \langle v_0,\dots,v_p\rangle$.
Since $E \subseteq A \times N$, every family of
$\beta$-specifications describing a natural transformation
as in \eqref{eq:alpha-bar-nat} is equivalent to a set
\vspace{-.5ex}\[\vspace{-.5ex}
	\left\{
	\frac{
		\Phi_{m_n,b_n} = \emptyset \quad
		\phi_{j} \in \Phi_{l_j,a_j} \quad 
		\totalweight{\phi_j} = v_j \quad  
		\corestr{\phi_j}{\mathfrak C_i} \ni y_i
	}{
		\psi(z_1,\dots,z_{|\psi|_*}) \in 
		\bar\alpha_{\mathtt f,c}(\langle \Phi_1,\dots,\Phi_{|\mathtt f|}\rangle)
	}\,\middle|
		\begin{array}{l}\scriptstyle
			\langle m_n,b_n\rangle \neq \langle l_j,a_j\rangle,\\\scriptstyle
			z_k \in \{\phi_j[\zeta_j], y_i\}+N,\\
			\psi
		\end{array}\!\!\!
	\right\}_\text{im.fin.}
\]
containing finitely many rules for every $E$ and $\vec v$.
This set corresponds to an $\alpha$-specification i.e.~an image-finite
set like the following:
\vspace{-.5ex}\[\vspace{-.5ex}
	\left\{
	\frac{
		\Phi_{m_n}(b_n) = \emptyset \quad
		\phi_{j} \in \Phi_{l_j}(a_j) \quad 
		\totalweight{\phi_j} = v_j \quad  
		\corestr{\phi_j}{\mathfrak C_i} \ni y_i
	}{
		\psi(z_1,\dots,z_{|\psi|_*}) \in 
		\alpha_{\mathtt f,c}(\langle \langle x_1,\Phi_1\rangle,\dots,\langle x_{|\mathtt f|},\Phi_{|\mathtt f|}\rangle\rangle)
	}\,\middle|
		\begin{array}{l}\scriptstyle
			\langle m_n,b_n\rangle \neq \langle l_j,a_j\rangle,\\\scriptstyle
			z_k \in \{\phi_j[\zeta_j], y_i, x_h\},\\\scriptstyle
			\psi
		\end{array}\!\!\!
	\right\}_\text{im.fin.}
\]
Finally, every family of $\alpha$-specifications equivalent to a natural
transformation as $\rho$ corresponds to an image-finite set of WF-GSOS rules
and an interpretation. Therefore we conclude that for any $\rho$
there exist $\mathcal R$ and $\xi$ as in Figure~\ref{fig:diag-comp}.

\section{Conclusions and future work}\label{sec:concl}
\looseness=-1
In this paper we have presented WF-GSOS, a GSOS-style format for
specifying non-deterministic systems with quantitative aspects. A
WF-GSOS specification is composed by a set of rules for the derivation
of judgements of the form $P \xrightarrowu{a} \psi$, where $\psi$ is a
term of a specific signature, together with an \emph{interpretation} for these
terms as weight functions.  We have shown that a specification in this
format defines an ULTraS, and it is expressive enough to subsume other
more specific formats such as Klin's \emph{Weighted GSOS} for WLTS
\cite{ks2013:w-s-gsos}, and Bartel's \emph{Segala-GSOS} for Segala
systems \cite[§5.3]{bartels04thesis}, and those subsumed by them
e.g.~Klin and Sassone's Stochastic GSOS \cite{ks2013:w-s-gsos} and
Bloom's GSOS \cite{bloomIM:95}.  WF-GSOS induces naturally a notion of
(strong) bisimulation, which we have compared with
$\mathcal{M}$-bisimulation used in ULTraS.  We have also provided a
general categorical presentation of ULTraSs as coalgebras of a precise
class of functors, parametric on the underlying weight structure.
This presentation allows us to define categorically the notion of
\emph{abstract GSOS} for ULTraS, i.e., natural transformations of a
precise type.  We have proved that WF-GSOS specification format is
\emph{adequate} (i.e., sound and complete) with respect to this
notion.  Taking advantage of Turi-Plotkin's bialgebraic framework, we
have proved that the bisimulation induced by a WF-GSOS is always a
congruence; hence our specifications can be used for compositional and
modular reasoning in quantitative settings (e.g., for ensuring
performance properties). Moreover, the format is at least as
expressive as every GSOS specification format for systems subsumed by
ULTraS.

\paragraph{Related works}
In this paper we have shown that commutative monoids are enough to
define ULTraSs, their homomorphisms and bisimulations.  The original
work \cite{denicola13:ultras} assumed weights to be organised into a
partial order with bottom $(W,\leq,\bottom)$, but the order plays no
r\^ole in the definition besides distinguishing the point $\bottom$
used to express unreachability.  A monoidal sum is eventually and
implicitly assumed by the notion of $\mathcal{M}$-bisimulation and,
because of the definition of $M$-function, this operation is assumed
to be be monotone in both its components and to have $\bottom$ as
unit.  In other words, $\mathcal{M}$-bisimulation implicitly assumes
weights to form a commutative positively ordered monoid
$(W,+,\leq,0)$.  Any such a monoid is positive and hence it has a
natural order $a \trianglelefteq b \iff \exists c .\, a+ c = b$; this
order is the weakest one rendering the monoid $W$ positively ordered,
in the sense that for any such ordering $\leq$, it is
${\trianglelefteq} \subseteq {\leq}$.

\looseness=-1
We note that in \cite{denicola13:ultras}, weights used to define
ULTraSs are decoupled from those of $M$-functions; e.g., the formers
can be in $([0,1],\leq,0)$ and the latters in $(\mathbb R^+_0,+,0)$.
However, the notion of constrained ULTraS is sill needed to precisely
capture probabilistic systems or, in other words, the use of partial
orders may still require to embed the systems under study into a
larger class of ULTraSs.  We remark that $(\mathbb R^+_0,+,0)$ is the
smaller completion of $([0,1],\leq,0)$ under $+$ and in this sense the
embedding can be seen as canonical.  Therefore, defining ULTraSs in
terms of commutative monoids is a conservative generalisation that
additionally provides a natural notion of homomorphisms and hence
bisimulations.  As a side note, existence of bottoms does not allow
weights to have opposites, e.g., to model opposite transitions like in
calculi for reversible computations.

Although in this paper we have taken ULTraSs as a reference,
WF-GSOS can be interpreted in other meta-models, such as FuTSs
\cite{latella:qapl2015}.  Like ULTraSs, FuTSs have state-to-function
transitions, but admit several distinct domains for weight functions
and more free structure besides the strict alternation between
non-deterministic and quantitative steps. In their more general
form, they can be understood as coalgebras for functors of shape:
\begin{equation}
	\label{eq:futs-fun}
	F_{\vec{A},\vec{\mathfrak W}} = 
	(\f{\mathfrak W_{0,k_0}}\dots\ff{\mathfrak W_{1,0}})^{A_0}
	\times\dots
	(\f{\mathfrak W_{n,k_n}}\dots\ff{\mathfrak W_{1,0}})^{A_n}
\vspace{-1ex}
\end{equation}
where each $\mathfrak W_{i,j}$ in $\vec{\mathfrak W}$ is a commutative monoid and each $A_i$ in $\vec{A}$ is a set.
We remark that, although in \cite{latella:qapl2015} weights are drawn 
from semirings, commutative monoids are sufficient to define $\fw$ and hence define FuTS, homomorphisms and eventually bisimulations. Moreover, Lemma~\ref{lem:wpp}
readily generalises to \eqref{eq:futs-fun}: if weights are drawn only
from positive refinement monoids then any such functor is wpp.
No rule format for FuTSs has been published yet; 
we believe the WF-GSOS specification format to be a step in this direction 
because of the similarities between the behavioural functors involved.
This would allow us to formulate compositionality results for
(meta)calculi defining FuTSs, e.g., the framework for stochastic 
calculi proposed in \cite{denicola13:ustoc}. Indeed
since ULTraSs can be viewed as FuTSs 
(assuming commutative monoids as a common ground)
any specification format for the latter that is both correct and complete 
w.r.t.~the suitable abstract GSOS law will necessarily subsume WF-GSOS.

The systems considered in this paper can be seen as generalised Segala
systems. We showed how the proposed format subsumes Bartels'
Segala-GSOS; however, this is not the only specification format for
this kind of systems.  In \cite{gdl2012:treerules} Gebler et
al.~proposed a $nt\mu f\nu/nt\mu x\nu$ rule format for describing
Segala systems. Since Turi-Plotkin seminal paper \cite{tp97:tmos} it
is well known that GSOS and coGSOS (i.e~tree-rule formats such as that
in \cite{gdl2012:treerules}) correspond to distributive laws of
completely different shapes: the former distribute monads over
copointed endofunctors whereas the latter distribute pointed
endofunctors over comonads.  These different shapes have obvious
implications on the data available to the derivation rules: monads
provide views ``inside terms'' whereas comonads provide views ``inside
executions''.  Their common generalisation are laws distributing
monads over comonads but has limited practical benefits because it
does not translate to any concrete rule format that would be complete
for any specification containing both GSOS and coGSOS
\cite{klin:sos2014}.

\paragraph{Future work}
The categorical characterization of ULTraS systems paves the way for
further interesting lines of research. One is to develop
Hennessy-Milner style modal logics for quantitative systems at the
generality level of the ULTraS framework. In fact, Klin has shown in
\cite{klin09:sosmlogic} that HML and CCS are connected by a
(contravariant) adjunction.  A promising direction is to follow this
connection taking advantage of the bialgebraic presentation of ULTraSs
provided in this paper. Another is to explore the implications of the
recent developments in the coalgebraic understanding of unobservable 
moves \cite{bmp:arxiv14-unobs,bonchi2015killing} in the context of this work.
An intermediate step in this direction is to develop a suitable monad structure
for $\pf\fw$ which is, in general, not a monad (cf.~$\pf\mathcal{D}$ where $\mathcal{D}$ is the probability distribution monad). This alone will
allow us to define e.g.~trace and testing equivalences in a principled coalgebraic way.

\vspace{-1ex}
\paragraph{Acknowledgements} We thank Rocco De Nicola, Daniel Gebler,
the anonymous reviewers and the QAPL'14 participants for useful discussions on the conference
version of this paper.
This work is partially supported by MIUR PRIN project 2010LHT4KM, \emph{CINA}.

\ifappendix
\clearpage
\appendix

\section{Weighted transition  Systems}
\label{apx:vs-wlts}

Weighted labelled transition systems 
(e.g.~\cite{handbook:weighted2009,ks2013:w-s-gsos}) 
are LTS whose transition are assigned a weight drawn from a commutative monoid
$\mathfrak W  =(W,0,+)$. Henceforth we will write $\mathfrak W$-LTS for $\mathfrak W$-Weighted LTS
or in general WLTS if no specific monoid is intended.
\begin{definition}[{\cite[Def.~2]{ks2013:w-s-gsos}}]
\label{def:wlts}
  Given a commutative monoid $\mathfrak W = (W,+,0)$, a \emph{$\mathfrak W$-weighted LTS} 
  is a triple $(X,A,\rho)$ where:
  \begin{itemize}
     \item $X$ is a set of \emph{states} (processes);
     \item $A$ is a set of \emph{labels} (actions);
     \item $\rho:X\times A \times X \to W$ is a \emph{weight
         function}, mapping each triple of $X\times A \times X$ to a
       weight.
  \end{itemize}
  $(X,A,\rho)$ is said to be \emph{image-finite} iff 
  for each $x\in X$ and $a\in A$, the set $\{ y \in X \mid \rho(x,a,y)\neq 0\}$
  is finite.
\end{definition}
It is well-known that, for suitable choices of $\mathfrak W$ and
constraints, WLTS subsume several kind of systems such as:
\begin{itemize}
\item $(\{\ltrue,\lfalse\},\lor,\lfalse)$ for non-deterministic systems;
\item $(\mathbb R^+_0,+,0)$ for rated systems \cite{ks2013:w-s-gsos,denicola09:rts} (e.g.~CTMCs);
\item $(\mathbb R^+_0,+,0)$ and 
  $\forall x \in X,a \in A\  \sum_{y\in X}\rho(x,a,y) \in \{0,1\}$ for
  generative (or fully) probabilistic systems;
\item $(\mathbb R^+_0,+,0)$ and 
  $\forall x \ \sum_{a\in A,y\in X}\mathrm P(x,a,y) \in \{0,1\}$ for
  reactive probabilistic systems;
\item $(\mathbb R^+_0,\max,0)$ for ``capabilities'' (weights denotes the capabilities of a process and similar capabilities add up to a stronger one);
\item etc.
\end{itemize}
Moreover, Klin defined in \cite{ks2013:w-s-gsos} a notion for
WLTS (based on cocongruences) which uniformly instantiates
to known bisimulations for systems expressible in the WLTS framework.
\begin{definition}[{\cite[Def.~4]{ks2013:w-s-gsos}}]
\label{def:w-bisim}
  Given two $\mathfrak W$-LTS $(X,A,\phi)$ and $(Y,A,\psi)$, 
  a \emph{$\mathfrak W$-bisimulation} is a relation $R \subseteq X \times Y$ 
  s.t.~for each pair $(x,y) \subseteq X \times Y$, 
  $(x,y) \in R$ implies that for each $a\in A$
  and each $(C,D)$ of $R^\star$:
  \[\sum_{c\in C}\phi(x,a,c) = \sum_{c\in D}\psi(y,a,d)\text.\]
\end{definition}

WLTS are precisely functional ULTraS and, as stated in
Proposition~\ref{prop:w-bisim}, every weighted bisimulation
for a WLTS is a bisimulation for the corresponding functional ULTraS and vice versa.
\begin{proof}[Proof of Proposition~\ref{prop:w-bisim}]
\label{proof:w-bisim}
Trivially, there is a 1-1 correspondence between $\mathfrak W$-LTS 
and functional $\mathfrak W$-ULTraS. Then, Proposition~\ref{prop:w-bisim}
readily follows by observing that, for any given pair of
$\mathfrak W$-LTS/ULTraS $(X,A,\rightarrowu_X)$ $(Y,A,\rightarrowu_Y)$, Definition~\ref{def:bisim}
degenerates in Definition~\ref{def:w-bisim} because 
for any $x \xrightarrow{a} \phi$ there is exactly one
$y \xrightarrow{a} \phi$.
\end{proof}

The coalgebraic understanding of WLTS makes the correspondence even more
immediate. In fact, there exists a bijective map between $\mathfrak W$-LTS
with labels in $A$ and $(\fw)^A$-coalgebras (cf.~\cite[Prop.~8]{ks2013:w-s-gsos})
and every $\mathfrak W$-weighted bisimulation arise from a cocongruence (cf.~\cite[Prop.~9]{ks2013:w-s-gsos}). Then, consider the natural transformations
$F : (\fw)^A \Longrightarrow (\pf\fw)^A$
and $G = (\pf\fw)^A \Longrightarrow (\fw)^A$:
\[
F_X(\phi)(a) \defeq \{\phi(a)\} \qquad
G_X(\Phi)(a) \defeq \lambda x: X.\sum_{\rho \in \Phi(a)}\rho(x)
\]
which lift the $\mathfrak W$-LTS behaviour to ULTraS and back.
These extends by composition to the functors, $\widetilde{F}$ and $\widetilde G$,
between the categories of coalgebras for $(\fw)^A$ and $(\pf\fw)^A$.
\[\begin{tikzpicture}[auto]
  \node (n0) at (0,0) {\((\fw)^A\cat{-CoAlg}\)};
  \node (n1) at (0,-2) {\((\pf\fw)^A\cat{-CoAlg}\)};
  \draw[->,bend right] (n0) to node[swap]{\(\widetilde F\)} (n1);
  \draw[->,bend right] (n1) to node[swap]{\(\widetilde G\)} (n0);
\end{tikzpicture}\]
The two are not adjoint but, the former is faithful and injective on objects 
whereas the latter is full and surjective on objects. 
Moreover $\widetilde G$ preserves the final coalgebra.

The natural transformations $F$ and $G$ give rise to the arrows
$G\hole F$ and $F\hole G$ (pictured below) by pre- and post- composition
and such that the first is injective and the second is surjective.
\[\begin{tikzpicture}[auto]
  \node (n0) at (0,0) {\(\cat{Nat}(\Sigma(\id \times (\fw)^A,(\fw T^\Sigma)^A)\)};
  \node (n1) at (0,-2) {\(\cat{Nat}(\Sigma(\id \times (\pf\fw)^A,(\pf\fw T^\Sigma)^A)\)};
  \draw[->,bend right] (n0) to node[swap]{\(G\hole F\)} (n1);
  \draw[->,bend right] (n1) to node[swap]{\(F\hole G\)} (n0);
\end{tikzpicture}\]

The functors above prove that on the same monoid ULTraS are a strict superclass of
WLTS. By a quick cardinality reasoning it is possible to extend the inclusion
result to the case where the monoid is allowed to change. In fact, for any $\mathfrak W = 
(W,+,0)$ s.t.~$|W| > 1$ there is no monoid $\mathfrak V = (V,\cdot,1)$ such that 
$2^{(|W|^x)} = |V|^x$. 
Unfortunately we cannot rule out the possibility of
``determinizing'' every ULTraS to some WLTS while preserving and reflecting
behavioural equivalences.

\section{Segala Systems}
\label{apx:vs-segala}
In their general format, Segala systems \cite{sl:njc95} are state 
machines (originally introduced as automata) 
whose transitions can be pictured as being
made of two steps belonging to two different behavioural
aspects: the first sub-step is non-deterministic and
the second one is probabilistic.
The following definitions are taken from \cite{sl:njc95}
with minor notational differences and by restricting to finite
probability distributions (whereas the original definition is given
to discrete at most countable probability spaces) for conciseness
and uniformity with the restriction to image-finite systems made in the
paper.
\begin{definition}
\label{def:segala-ts}
  A \emph{Segala system} is a triple $(X,A,\rightarrowu)$ where:
  \begin{itemize}
     \item $X$ is a set of \emph{states} (processes);
     \item $A$ is a set of \emph{labels} (actions);
     \item ${\rightarrowu} \subseteq X\times A \times \mathcal D(X)$ 
       a \emph{transition relation}
       between states and discrete probability spaces over pairs of labels and states.
  \end{itemize}
\end{definition}

\begin{definition}[{\cite[Def.~14]{sl:njc95}}]
\label{def:segala-bisim}
  Let $(X,A,\rightarrowu_X)$ and $(Y,A,\rightarrowu_Y)$ be two
  Segala systems.
  A bisimulation is a relation $R \subseteq X \times Y$ 
  such that for each $(x,y) \in X \times Y$, $(x,y) \in R$ 
  implies that for each $a\in A$
  for each $x \xrightarrowu{a} \phi$ there is $y\xrightarrowu{a} \psi$
  s.t.~for each $(C,D) \in R^\star$
  $\sum_{c\in C}\phi(c) = \sum_{d\in D}\psi(d)$
  and symmetrically for $y$.
\end{definition}

\begin{proof}[Proof of Proposition~\ref{prop:segala-bisim}]
\label{proof:segala-bisim}
Clearly $\mathcal DX \subseteq \f{\mathbb R^+_0}X$ and hence Segala systems
are constrained ULTraS. Then, Proposition~\ref{prop:segala-bisim}
readily follows by observing that the two notions coincide on
the non-deterministic part and then on the summations over
elements of $R^*$ i.e.~the extension of $R$
\end{proof}

\printproofs

\fi 

\end{document}